\documentclass{article}
\usepackage{lmodern}
\usepackage[margin=1in]{geometry}
\usepackage[utf8]{inputenc} % allow utf-8 input
\usepackage[T1]{fontenc}    % use 8-bit T1 fonts
\usepackage{url}            % simple URL typesetting
\usepackage{booktabs}       % professional-quality tables
\usepackage{amsfonts}       % blackboard math symbols
\usepackage{nicefrac}       % compact symbols for 1/2, etc.
\usepackage{microtype}      % microtypography
\usepackage{natbib}

\usepackage{amsmath}
\usepackage{amsthm}
\usepackage{color}
\usepackage{graphicx}
\usepackage{setspace}
\usepackage{esint}

\usepackage[unicode=true,bookmarks=true,bookmarksnumbered=false,bookmarksopen=false,breaklinks=false,pdfborder={0 0 0},pdfborderstyle={},backref=page,colorlinks=true]{hyperref}
\hypersetup{pdftitle={Estimating Convergence},pdfauthor={Biswas Jacob Vanetti},linkcolor=RoyalBlue,citecolor=RoyalBlue}
\usepackage[dvipsnames,svgnames,x11names,hyperref]{xcolor}

\newtheorem{Th}{Theorem}[section]

\newtheorem{Prop}[Th]{Proposition}
\newtheorem{Asm}[Th]{Assumption}

\newtheorem{Def}[Th]{Definition}

\newcommand{\Bernoulli}{\mathrm{Bernoulli}}

\usepackage[ruled]{algorithm2e}
\usepackage{graphicx}
\usepackage{graphics}

\title{ Estimating Convergence of Markov chains with \( L \)-Lag Couplings }

\author{
	Niloy Biswas \thanks{Department of Statistics, Harvard University, Cambridge, USA. Email: niloy\_biswas@g.harvard.edu}
  	\and 
  	Pierre E. Jacob \thanks{Department of Statistics, Harvard University, Cambridge, USA. Email: pjacob@fas.harvard.edu}
  	\and 
  	Paul Vanetti \thanks{Department of Statistics, University of Oxford, Oxford, UK. Email: paul.vanetti@spc.ox.ac.uk}
}

\begin{document}

\maketitle

\begin{abstract}
Markov chain Monte Carlo (MCMC) methods generate samples that are
asymptotically distributed from a target distribution of interest as the number
of iterations goes to infinity. Various theoretical results provide upper
bounds on the distance between the target and marginal distribution after a
fixed number of iterations. These upper bounds are on a case by case basis and
typically involve intractable quantities, which limits their use for
practitioners. We introduce \(L\)-lag couplings to generate computable,
non-asymptotic upper bound estimates for the total variation or the Wasserstein
distance of general Markov chains. We apply \(L\)-lag couplings to the tasks of
(i) determining MCMC burn-in, (ii) comparing different MCMC
algorithms with the same target, and (iii) comparing exact and
approximate MCMC. Lastly, we (iv) assess the bias of sequential Monte Carlo and
self-normalized importance samplers.
\end{abstract}

% \tableofcontents

\section{Introduction}

Markov chain Monte Carlo (MCMC) algorithms generate Markov chains that are invariant with respect
to probability distributions that we wish to approximate.
Numerous works help understanding the convergence
of these chains to their invariant distributions, hereafter denoted by $\pi$.  Denote by $\pi_t$
the marginal distribution of the chain $(X_t)_{t\geq 0}$ at time $t$.
The discrepancy between $\pi_t$ and $\pi$ can be measured in different ways,
typically the total variation (TV) distance or the Wasserstein distance in the MCMC literature.
Various results provide upper bounds on this distance,
of the form $C(\pi_0) f(t)$, where $C(\pi_0)<\infty$ depends on $\pi_0$
but not on $t$, and where $f(t)$ decreases to zero as $t$ goes to infinity, typically geometrically;
see Section 3 in \citet{roberts_rosenthal_2004} for a gentle survey,
and \citet{durmus2016subgeometric,dalalyan_2017,dwivedi_chen_wainwright_yu_2018} 
for recent examples.
These results typically relate convergence rates 
to the dimension of the state space or to various features of the target. 
Often these results do not provide 
computable bounds on the distance between $\pi_t$ and $\pi$,
as $C(\pi_0)$ and $f(t)$ typically feature unknown constants; although see 
\citet{rosenthal1996analysis} where these constants can be bounded analytically,
and \citet{cowles1998simulation} for examples where they can be numerically approximated.

Various tools have been developed to assess the quality 
of MCMC estimates. Some focus on the behaviour of the chains assuming stationarity,
comparing averages computed within and across chains, or defining various notions of effective sample sizes
based on asymptotic variance estimates 
(e.g.\ \citet{gelman_rubin_1992,geweke_1992,brooks_gelman_1998,vats_flegal_jones_2019},
\citet[Chapter 8]{robert_casella_2013}).
Few tools provide computable bounds on the distance
between $\pi_t$ and $\pi$ for a fixed $t$; some are mentioned in \citet{brooks1998assessing}
for Gibbs samplers with tractable transition kernels.
Notable exceptions, beyond \citet{cowles1998simulation} mentioned above, include the method of 
\citet{valen_johnson_1996,valen_johnson_1998} which relies on coupled Markov chains.
A comparison with our proposed method will be given in Section \ref{sec:johnson}. 

We propose to use \(L\)-lag couplings of Markov chains to estimate
the distance between $\pi_t$ and $\pi$ for a fixed time $t$,
building on $1$-lag couplings used to obtain unbiased estimators
in \citet{glynn_2014,jacob_2019}.
The discussion of \citet{jacob_2019} mentions that upper bounds on the TV
between $\pi_t$ and $\pi$ can be estimated with such couplings.
We generalize this idea to \(L\)-lag couplings, which provide sharper bounds, particularly for small values of $t$. 
The proposed technique extends to a class of probability metrics \citep{sriperumbudur2012empirical} 
beyond TV. 
We demonstrate numerically that the bounds 
provide a practical assessment of convergence for 
various popular MCMC algorithms, on either discrete or continuous and possibly high-dimensional spaces.
The proposed bounds can be used to (i) determine burn-in
period for MCMC estimates, to (ii) compare different MCMC algorithms targeting the same distribution, or
to (iii) compare exact and approximate MCMC algorithms, such as Unadjusted and Metropolis-adjusted 
Langevin algorithms, providing a computational companion to studies such as \citet{dwivedi_chen_wainwright_yu_2018}.
We also (iv) assess the bias of sequential Monte Carlo and self-normalized importance samplers.

In Section \ref{section:L_lag_couplings} we
introduce \(L\)-lag couplings to estimate metrics between marginal and invariant distributions
of a Markov chain. We illustrate the method on simple examples,
discuss the choice of $L$, and compare with the approach of \citet{valen_johnson_1996}. In Section \ref{section:experiments} we
consider applications including Gibbs samplers on the Ising model and gradient-based MCMC algorithms on log-concave targets. 
In Section \ref{section:SMC_bound} we assess the bias of sequential Monte Carlo and self-normalized importance samplers.
All scripts in R are available at \url{https://github.com/niloyb/LlagCouplings}. 

\section{\texorpdfstring{$L$}{L}-lag couplings} \label{section:L_lag_couplings}

Consider two Markov chains \((X_t)_{t \geq 0}\), \((Y_t)_{t \geq 0}\), 
each with the same
initial distribution \(\pi_0\) and Markov kernel \(K\) on
\((\mathbb{R}^d, \mathcal{B}(\mathbb{R}^d))\) which is \(\pi\)-invariant. 
Choose some integer \(L \geq 1 \) as
the lag parameter. 
We generate the two chains using Algorithm \ref{main_algo_coupling}.
The joint Markov kernel \(\bar{K}\) 
on \((\mathbb{R}^d \times \mathbb{R}^d, \mathcal{B}(\mathbb{R}^d
\times \mathbb{R}^d))\) is such that, 
for all $x$, $y$, \(\bar{K}((x,y), (\cdot, \mathbb{R}^d)) = K(x, \cdot)\), and \(\bar{K}((x,y),
(\mathbb{R}^d, \cdot)) = K(y, \cdot)\). 
This ensures that $X_t$ and $Y_t$ have the same marginal distribution at all times $t$.
Furthermore, \(\bar{K}\) is constructed such that the pair of chains can meet exactly
after a random number of steps, i.e.\ the meeting time\ 
$\tau^{(L)} := \inf \{ t > L : X_t = Y_{t-L} \}$
is almost surely finite. Finally we assume that the chains remain faithful after meeting, i.e.\ \( X_t = Y_{t-L} \) for all \( t\geq \tau^{(L)} \).

%Consider two Markov chains \((X_t)_{t \geq 0}\), \((Y_t)_{t \geq 0}\), 
%each with the same
%initial distribution \(\pi_0\) and Markov kernel \(K\) on
%\((\mathbb{R}^d, \mathcal{B}(\mathbb{R}^d))\) which is \(\pi\)-invariant. 
%Introduce a joint Markov transition kernel
%\(\bar{K}\) on \((\mathbb{R}^d \times \mathbb{R}^d, \mathcal{B}(\mathbb{R}^d
%\times \mathbb{R}^d))\) such that, 
%for all $x$, $y$, \(\bar{K}((x,y), (\cdot, \mathbb{R}^d)) = K(x, \cdot)\), and \(\bar{K}((x,y),
%(\mathbb{R}^d, \cdot)) = K(y, \cdot)\). Choose some integer \(L \geq 1 \) as
%the lag parameter. 
%We generate the two chains as follows:
%first sample $X_0\sim \pi_0$, $X_t|X_{t-1} \sim K(X_{t-1},\cdot)$
%for $t=1,\ldots,L$, and sample $Y_0 \sim \pi_0$.
%For $t>L$, sample \((X_t, Y_{t-L})|(X_{t-1}, Y_{t-L-1})
%\sim \bar{K}((X_{t-1},Y_{t-L-1}),\cdot)$.
%The construction ensures that $X_t$ and $Y_t$ have the same marginal distribution at all times $t$.
%Furthermore, we will choose the joint kernel \(\bar{K}\) such that the pair of chains can meet exactly
%after a random number of steps, i.e.\ the meeting time\ 
%$\tau^{(L)} := \inf \{ t > L : X_t = Y_{t-L} \}$
%is almost surely finite. Finally we assume that the chains remain faithful after meeting, i.e.\ \( X_t = Y_{t-L} \) for all \( t\geq \tau^{(L)} \).

Various constructions for $\bar{K}$ have been derived in the literature:
for instance coupled Metropolis-Hastings and Gibbs kernels in
\citet{valen_johnson_1996,jacob_2019}, coupled Hamiltonian Monte Carlo kernels in \citet{mangoubi2017rapid,bou2018coupling,heng_2018},
and coupled particle Gibbs samplers in \citet{chopin2015particle,andrieu2018uniform,Jacob_Lindsten_Schon_2019}.
%Given $\pi_0$, $K$, $\bar{K}$ and a lag $L$, Algorithm \ref{main_algo_coupling} summarizes how to sample an \(L\)-lag coupling of a pair of chains up to the meeting time \( \tau^{(L)} \). 

\begin{algorithm}[H] \caption{Sampling \(L\)-lag meeting times}
\DontPrintSemicolon
\textbf{Input:} lag \( L \geq 1\), initial distribution \(\pi_0\), single kernel \( K \) and joint kernel \( \bar{K} \) \;
\textbf{Output:} meeting time \( \tau^{(L)} \), and chains \( (X_t)_{0 \leq t \leq \tau^{(L)} }, (Y_t)_{0 \leq t \leq \tau^{(L)}-L } \) \;
Initialize:  generate \(X_0 \sim \pi_0\), \( X_t|X_{t-1} \sim  K(X_{t-1}, \cdot)\) for \(t=1,\ldots,L\), and \(Y_0 \sim \pi_0\) \;
 \For{\( t > L \)}{
 Sample \( (X_t, Y_{t-L})| (X_{t-1}, Y_{t-L-1}) \sim \bar{K}((X_{t-1}, Y_{t-L-1}), \cdot) \) \;
\lIf{ \( X_{t} = Y_{t-L} \) }{
   \Return \( \tau^{(L)} := t \), and chains \( (X_t)_{0 \leq t \leq \tau^{(L)} },  (Y_t)_{0 \leq t \leq \tau^{(L)}-L } \)
   }
}
 \label{main_algo_coupling}
\end{algorithm}

%The main contribution of this article is to show how Algorithm \ref{main_algo_coupling}
%can be used to upper bound various distances between $\pi_t$ and $\pi$.

We next introduce integral probability metrics (IPMs, e.g.\ \citet{sriperumbudur2012empirical}).

\begin{Def} (Integral Probability Metric). Let \(\mathcal{H} \) be a class of real-valued functions on a measurable space \( \mathcal{X} \). For all probability measures \( P, Q \) on  \( \mathcal{X} \), the corresponding IPM is defined as:
\begin{equation} \label{def:IPM}
d_{\mathcal{H}}(P, Q) := \sup\limits_{h \in \mathcal{H}} \Big| \mathbb{E}_{X \sim P} [h(X)] -  \mathbb{E}_{X \sim Q} [h(X)] \Big|.
\end{equation}

%With \( \mathcal{H}:= \{ h: \sup_{x \in \mathcal{X}} |h(x)| \leq 1 \} \) 
%and multiplicative factor of $1/2$, $d_{\mathcal{H}}$ is referred to as the total variation distance
%and denoted by $d_{TV}$.
%With $\mathcal{H}= \{ h: |h(x)- h(y)| \leq \|x-y\|_{1} \ \forall x, y \in \mathbb{R}^d \}$,
%$d_{\mathcal{H}}$ is referred to as the 1-Wasserstein distance and denoted by $d_W$. Here $\|\cdot \|_1$ refers to the $L_1$ norm on $\mathbb{R}^d$.
\end{Def}

Common IPMs include total variation distance $d_{\text{TV}}$ with \( \mathcal{H}:= \{ h: \sup_{x \in \mathcal{X}} |h(x)| \leq 1/2 \} \), and 1-Wasserstein distance $d_{\text{W}}$ with $\mathcal{H}= \{ h: |h(x)- h(y)| \leq d_{ \mathcal{X}}(x, y), \ \forall x, y \in \mathcal{X} \}$, 
where $d_{ \mathcal{X}}$ is a metric on $\mathcal{X}$ \citep{PeyreCuturi2019OptTransport}. Our proposed method applies to IPMs such that \( \sup_{h \in \mathcal{H}} | h(x) - h(y) |
\leq M_\mathcal{H}(x,y) \) for all \(x, y \in \mathcal{X} \), for some computable function $M_\mathcal{H}$ on  \( \mathcal{X} \times
\mathcal{X} \). For $d_{\text{TV}}$
we have $M_\mathcal{H}(x,y) = 1$, and for $d_\text{W}$ we have
$M_\mathcal{H}(x,y) = d_{ \mathcal{X}}(x, y)$.

%Our proposed method applies to IPMs such that we can compute 
%a function $M_\mathcal{H}$ on  \( \mathcal{X} \times
%\mathcal{X} \) such that \( \sup_{h \in \mathcal{H}} | h(x) - h(y) |
%\leq M_\mathcal{H}(x,y) \) for all \(x, y \in \mathcal{X} \). 
%For the total variation distance
%we have $M_\mathcal{H}(x,y) = 1$, and for the 1-Wasserstein distance
%$M_\mathcal{H}(x,y) = \|x-y\|_1$. In the following $\mathcal{H}$ refers
%to a class of functions such that we can compute $M_{\mathcal{H}}$.

%With a similar motivation for the assessment of sample approximations, and not
%restricted to the MCMC setting, \citep{gorham_mackey_2015} considers choices of
%function sets $\mathcal{H}$ that facilitate the computation of integral
%probability metrics.  Here we focus on the MCMC setting and directly aim at
%upper bounds on the total variation and Wasserstein distance.

With a similar motivation for the assessment of sample approximations, and not
restricted to the MCMC setting, \citet{gorham_mackey_2015} considers a restricted 
class of functions \(\mathcal{H}\) to develop a specific measure of sample quality based on
Stein's identity. \citet{liu2016kernelgoodnessoffit, chwialkowski_2016_kerneltest} combine Stein’s
identity with reproducing kernel Hilbert space theory to develop 
goodness-of-fit tests. \citet{gorham_mackey_2018} obtains further results
and draws connections to the literature on couplings of Markov processes. 
Here we directly aim at upper bounds on the total variation and
Wasserstein distance.
The total variation controls the maximal difference
between the masses assigned by $\pi_t$ and $\pi$ on any measurable set,
and thus directly helps assessing the error of histograms of the target marginals.
The 1-Wasserstein distance controls the error 
made on expectations of 1-Lipschitz functions, which 
with $\mathcal{X}=\mathbb{R}^d$ and $d_{\mathcal{X}}(x, y) = \| x - y \|_1$ 
(the $L_1$ norm on $\mathbb{R}^d$) include all first moments.

\subsection{Main results}
We make the three following assumptions similar to those of \citet{jacob_2019}.

\begin{Asm} \label{assumption_1} (Marginal convergence and moments.) For all $h\in \mathcal{H}$, as $t \rightarrow \infty$, $\mathbb{E}[h(X_t)] \rightarrow \mathbb{E}_{X \sim \pi} [h(X)]$. Also, $\exists \eta > 0, D < \infty$ such that $\mathbb{E}[M_{\mathcal{H}}(X_t, Y_{t-L})^{2+\eta}] \leq D$ for all $t \geq L$.
\end{Asm}
The above assumption is on the marginal convergence of the MCMC algorithm
and on the moments of the associated chains.
The next assumptions are on the coupling operated by the joint kernel $\bar{K}$.

\begin{Asm} \label{assumption_2} (Sub-exponential tails of meeting times.) The chains are such that the meeting time \( \tau^{(L)} := \inf \{ t > L : X_t = Y_{t-L} \} \) satisfies \( \mathbb{P}( \frac{\tau^{(L)}-L}{L} >t) \leq C \delta^t \) for all $t\geq 0$, for some constants $C<\infty$ and $\delta \in (0,1)$.
\end{Asm}
The above assumption can be relaxed to allow for polynomial tails as in \citet{middleton2018unbiased}.
The final assumption on faithfulness is typically satisfied by design.

\begin{Asm} \label{assumption_3} (Faithfulness.) The chains stay together after meeting: $X_t = Y_{t-L}$ for all $t \geq \tau^{(L)}$.
\end{Asm}

We assume that the three assumptions above hold in the rest of the article. 
The following theorem is our main result.

\begin{Th} \label{ipm_upper_bound} (Upper bounds.) For an IPM with function set $\mathcal{H}$ and upper bound $M_{\mathcal{H}}$, with the Markov chains \( (X_t)_{t \geq 0}, (Y_t)_{t \geq 0} \) 
    satisfying the above assumptions, for any $L\geq 1$, and any $t\geq 0$,
\begin{equation} \label{eq:ipm_upper_bound}
d_{\mathcal{H}}(\pi_t, \pi) \leq  \mathbb{E} \Big[ \sum_{j=1}^{ \bigl\lceil \frac{\tau^{(L)} - L -t}{L} \bigr\rceil }  M_{\mathcal{H}}(X_{t+jL}, Y_{t+(j-1)L}) \Big].
\end{equation}
\end{Th}
Here $\lceil x\rceil$ denotes the smallest integer above $x$, for $x\in \mathbb{R}$. 
When \( \lceil (\tau^{(L)} - L -t)/L \rceil \leq 0 \), the sum in inequality
\eqref{eq:ipm_upper_bound} is set to zero by convention. We next give a
short sketch of the proof. Seeing the invariant distribution $\pi$ as the limit 
of $\pi_t$ as $t\to\infty$, applying triangle inequalities, 
recalling that $d_{\mathcal{H}}(\pi_s, \pi_t) \leq \mathbb{E}[M_\mathcal{H}(X_s,X_t)]$ for all $s$, $t$,
we obtain 
\begin{flalign} \label{sketchproof:ipm_upper_bound}
d_{\mathcal{H}}(\pi_t, \pi) &\leq \sum_{j=1}^\infty d_{\mathcal{H}}(\pi_{t+jL}, \pi_{t+(j-1)L}) \leq \sum_{j=1}^\infty \mathbb{E}[M_{\mathcal{H}} (X_{t+jL}, X_{t+(j-1)L})].
%\nonumber \\
%&= \mathbb{E} \Bigg[ \sum_{j=1}^{ \bigl\lceil \frac{\tau^{(L)} - L -t}{L} \bigr\rceil } M_{\mathcal{H}}(X_{t+jL}, Y_{t+(j-1)L}) \Bigg].
\end{flalign}
The right-hand side of \eqref{eq:ipm_upper_bound} is retrieved by swapping expectation and limit, and noting that terms 
indexed by $j> \lceil (\tau^{(L)} - L -t)/L \rceil$ are equal to zero by Assumption \ref{assumption_3}.
The above reasoning highlights that increasing \(L\) leads to sharper bounds through the use of fewer triangle inequalities. 
A formal proof is given in the appendices.

Theorem \ref{ipm_upper_bound} gives the following bounds for $d_{\text{TV}}$ and $d_{\text{W}}$,
%\begin{Remark} \label{tv_upper_bound} (Total Variation upper bound) \( \mathcal{H}:= \{ h: \sup\limits_{x \in \mathcal{X}} |h(x)| \leq 1 \} \) gives 
\begin{align} 
    d_{\text{TV}}(\pi_t, \pi) &\leq  \mathbb{E} \Big[ \max(0, \bigl\lceil \frac{\tau^{(L)} - L -t}{L} \bigr\rceil) \Big],\label{eq:tv_upper_bound}\\
    d_{\text{W}}(\pi_t, \pi) &\leq   \mathbb{E} \Big[ \sum_{j=1}^{ \bigl\lceil \frac{\tau^{(L)} - L -t}{L} \bigr\rceil }  d_{\mathcal{X}}(X_{t+jL}, Y_{t+(j-1)L}) \Big] \label{eq:wasserstein_upper_bound}.
\end{align} 
%\end{Remark}
%\begin{Remark} \label{wasserstein_upper_bound} (Wasserstein distance upper bound) \( \mathcal{H}:= \{ h: \sup\limits_{x \in \mathcal{X}} |h(x)| \leq 1 \} \)  gives 
%\begin{equation} \label{eq:wasserstein_upper_bound}
%d_{\mathcal{W}_{\| \cdot \|}}(\pi_t, \pi) \leq   \mathbb{E} \Big[ \sum_{j=1}^{ \bigl\lceil \frac{\tau^{(L)} - L -t}{L} \bigr\rceil }  \Bigl\| X_{t+jL} - Y_{t+(j-1)L} \Bigl\|_{1} \Big]
%%\end{equation} for \( \| \cdot \|_{1} \) the \( L1 \)-norm on \( \mathbb{R}^d \).
%%\end{Remark}

For the total variation distance, the boundedness part of Assumption
\ref{assumption_1} is directly satisfied. 
For the 1-Wasserstein distance on $\mathbb{R}^d$ with 
$d_{\mathcal{X}}(x, y) = \| x - y \|_1$ (the $L_1$ norm on $\mathbb{R}^d$),
the boundedness part is equivalent to a
uniform bound of \( (2+\eta)\)-th moments of the marginal distributions for
some \( \eta>0 \).

We emphasize that the proposed bounds can be estimated directly by running 
Algorithm \ref{main_algo_coupling} $N$ times independently,
and using empirical
averages. All details of the MCMC algorithms and their couplings mentioned below are provided in the appendices. 

\subsection{Stylized examples}
\subsubsection{ A univariate Normal } \label{section:standard_normal}

We consider a Normal example where we can compute total variation and $1$-Wasserstein distances 
(using the $L_1$ norm on $\mathbb{R}$ throughout) exactly.
The target $\pi$ is \(\mathcal{N}(0,1) \) and the kernel $K$ is that of a
Normal random walk Metropolis-Hastings (MH) with
step size \( \sigma_{\text{MH}}=0.5 \). 
%given $X_{t-1}$, sample $X^\star \sim \mathcal{N}(X_{t-1},\sigma_{\text{MH}}^2)$
%and $U\sim \mathcal{U}(0,1)$ and set $X_{t} = X^\star$ if $U< \pi(X^\star)/\pi(X_{t-1})$; otherwise set $X_{t} = X_{t-1}$.
We set the initial distribution $\pi_0$ to be a point mass at $10$. 
The joint kernel $\bar{K}$ operates as follows. Given $(X_{t-1},Y_{t-L-1})$,
sample $(X^\star,Y^\star)$ from a maximal coupling of $p:= \mathcal{N}(X_{t-1},\sigma_{\text{MH}}^2)$
and $q:= \mathcal{N}(Y_{t-L-1},\sigma_{\text{MH}}^2)$. This is done using Algorithm \ref{max_coupling}, which  
ensures $X^\star \sim p$, $Y^\star \sim q$ and $\mathbb{P}(X^\star \neq Y^\star) = d_{\text{TV}}(p,q)$.
%That is, sample $X^\star \sim p$ and $W\sim\mathcal{U}(0,1)$. If $W<q(X)/p(X)$ set $Y^\star = X^\star$; otherwise sample $Y^\star \sim q$ and $W'\sim \mathcal{U}(0,1)$ until $p(Y^\star)/q(Y^\star)<W'$. One can check that $X^\star \sim p$, $Y^\star \sim q$, and $\mathbb{P}(X^\star \neq Y^\star) = d_{\text{TV}}(p,q)$. 
\begin{algorithm}[!htb]
\DontPrintSemicolon
\caption{A maximal coupling of \(p\) and \(q\)}
Sample \(X^* \sim p\), and \(W \sim \mathcal{U}(0,1)\)\;
\lIf{ \( p(X^*) W \leq q(X^*) \) }{
   set \(Y^*=X^*\) and \Return \( (X^*,Y^*) \)
   } \lElse {
   sample \( \tilde{Y} \sim q\) and \(\tilde{W} \sim \mathcal{U}(0,1) \) until  \( q(\tilde{Y}) \tilde{W} > p(\tilde{Y}) \). Set \(Y^*=\tilde{Y}\) and \Return \( (X^*,Y^*) \)
   }
\label{max_coupling}
\end{algorithm}

Having obtained $(X^\star, Y^\star)$, sample $U\sim \mathcal{U}(0,1)$; set $X_{t} = X^\star$ if $U< \pi(X^\star)/\pi(X_{t-1})$; otherwise set $X_{t} = X_{t-1}$. With the same $U$, set $Y_{t-L} = Y^\star$ if $U< \pi(Y^\star)/\pi(Y_{t-L-1})$; otherwise set $Y_{t-L} = Y_{t-L-1}$. Such a kernel $\bar{K}$ is a coupling of $K$ with itself, and Assumption \ref{assumption_3} holds by design. The verification of Assumption \ref{assumption_2} is harder but can be done via drift conditions in various cases; we refer to \citet{jacob_2019} for more discussion.

Figure \ref{TV_bounds_1}
shows the evolution of the marginal distribution of the chain, and the TV and $1$-Wasserstein 
distance upper bounds.
We use $L=1$ and $L=150$. For each \(L\), \(N=10000\) independent runs of Algorithm
\ref{main_algo_coupling} were performed to estimate the bounds in Theorem \ref{ipm_upper_bound} by empirical averages. 
Exact distances are shown for comparison. 
Tighter bounds are obtained with larger values of $L$, as discussed further in Section \ref{section:L_choice}. 

\begin{figure} [ht!]
\begin{center}
\includegraphics[width=0.65\paperwidth]{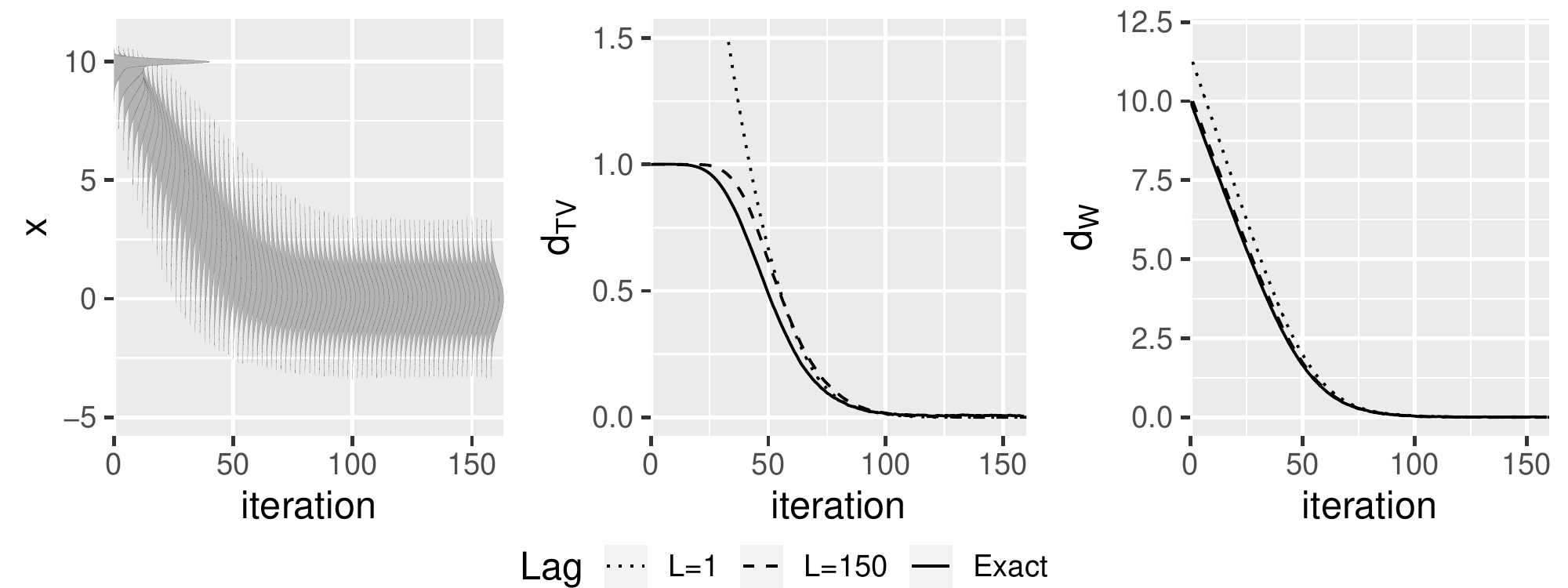}
\caption{Marginal distributions of the chain (left), and upper bounds on the total variation (middle) and the $1$-Wasserstein distance
    (right) between $\pi_t$ and $\pi$, for a Metropolis-Hastings algorithm targeting \(
    \mathcal{N}(0,1) \) and starting from a Dirac mass at $10$. With $L=150$
the estimated upper bounds for both are close to the exact distances.}
\label{TV_bounds_1}
\end{center}
\end{figure}

\subsubsection{ A bimodal target } \label{section:bimodal}
We consider a bimodal target to illustrate the limitations of the proposed technique. The target is 
$\pi  = \frac{1}{2}\mathcal{N}(-4, 1) + \frac{1}{2}\mathcal{N}(4, 1)$,
as in Section 5.1 of \citet{jacob_2019}. The MCMC algorithm is again random walk MH,
with \( \sigma_{\text{MH}}=1,  \pi_0 =
\mathcal{N}(10, 1)\). Now, the chains struggle to jump between the modes, 
as seen in Figure \ref{bimodal_combined} (left), which shows a histogram of the 500th marginal distribution 
from 1000 independent chains.  Figure \ref{bimodal_combined} (right) shows the TV upper bound
estimates for lags $L=1$ and $L=18000$ (considered very large),
obtained with $N\in \{1000, 5000, 10000\}$ independent runs of Algorithm \ref{main_algo_coupling}. 

With $L=18000$, we do not see a difference between the obtained upper bounds,
which suggests that the variance of the estimators is small for the different values of $N$.
In contrast, the dashed line bounds corresponding to lag \(L=1\) are
very different. This is because, over $1000$ experiments, the 1-lag meetings
always occurred quickly in the mode nearest to the initial distribution. However,
over $5000$ and $10000$ experiments, there were instances where one of the two chains jumped to the other mode
before meeting, resulting in a much longer meeting time. Thus the results obtained with $N=1000$
repeats can be misleading. This is a manifestation of the estimation error associated with empirical averages,
which are not guaranteed to be accurate after any fixed number $N$ of repeats. The shape of the bounds obtained with 
$L=18000$, with a plateau, reflects how the chains first visit one of the modes,
and then both.

%variance terms \( ( Var( \bigl\lceil \frac{\tau
%    ^{(L)} - L - t}{L} \bigr\rceil ) )_{t \geq 0} \) for \( L=1 \) are large.
%    The higher variance for lag \( L = 1 \) is a result of \( \tau^{(1)} \)
%    having fat tails in this example. \( \tau^{(1)} \) can take very large
%    values with a small probability (corresponding to cases when the two chains
%    are in two separate modes), and otherwise it take much smaller values (when
%    both chains couple at the \(+4\) mode near which they originally start).
%    For smaller sample sizes, this leads to estimation error for \(1\)-Lag
%    couplings and misleading bounds. One way to notice such estimation error is
%    to consider the estimated total variation bound at time \(0\), which
%    corresponds to the empricial estimate of \( \mathbb{E}\big[ \bigl\lceil
%    \frac{\tau ^{(L)} - L }{L} \bigr\rceil \big] \). Note in particular that
%    the bound for \(1\)-Lag is much larger than \(1\). This motivates the
%    discussion on Section \ref{section:L_choice} about the choice of lag \(L\).

\begin{figure} [ht!]
\begin{center}
\includegraphics[width=0.65\paperwidth]{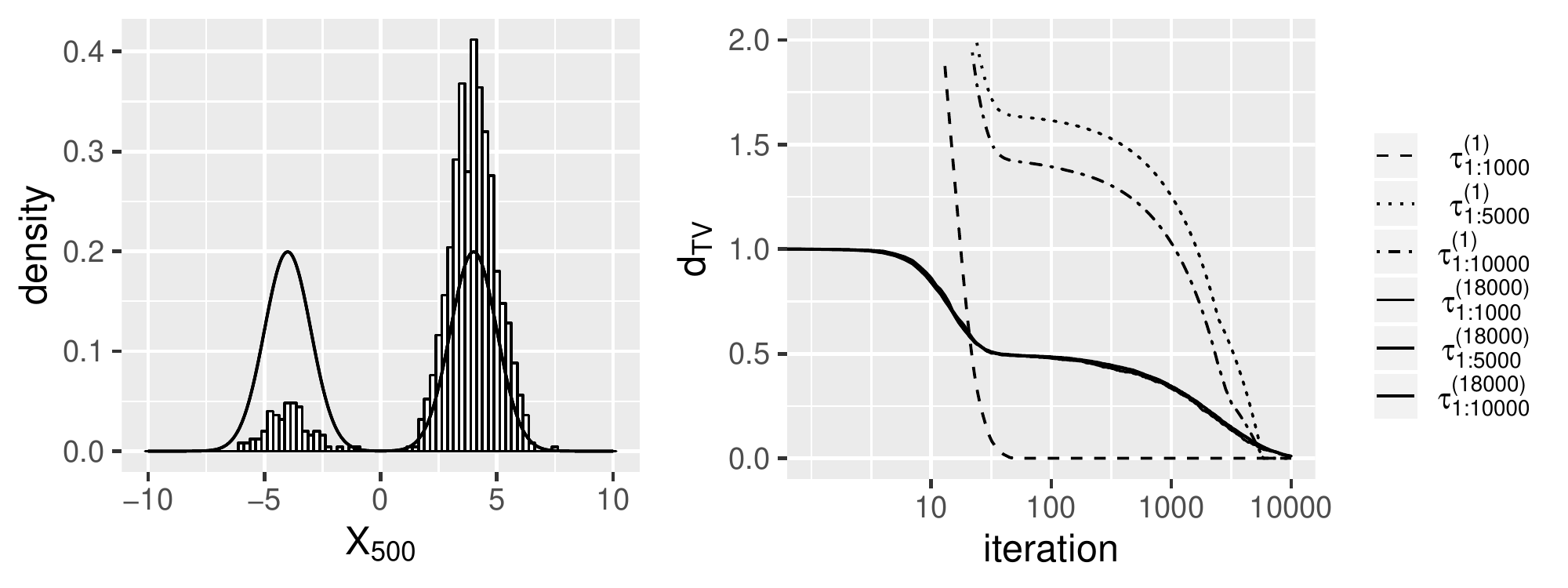}
\end{center}
\caption{ Metropolis-Hastings algorithm with \( \pi_0 \sim \mathcal{N}(10, 1), \sigma_{\text{MH}}=1 \) on a bimodal target. 
Left: Histogram of the 500th marginal distribution 
from 1000 independent chains, and target density in full line.
Right: Total variation bounds obtained with lags $L\in\{1, 18000\}$ and 
$N\in\{1000,5000,10000\}$ independent runs of Algorithm \ref{main_algo_coupling}. }
\label{bimodal_combined}
\end{figure}

\subsection{Choice of lag \texorpdfstring{$L$}{L}} \label{section:L_choice}

Section \ref{section:bimodal} illustrates the importance of the choice of lag
\(L\). Obtaining $\tau^{(L)}$
requires sampling $L$ times from $K$ and $\tau^{(L)}-L$ from $\bar{K}$. 
When $L$ gets large, 
we can consider $X_L$ to be at stationarity, while $Y_0$ still follows $\pi_0$.
Then the distribution of 
$\tau^{(L)}-L$ depends entirely on $\bar{K}$ and not on $L$. 
In that regime the cost of obtaining $\tau^{(L)}$ increases linearly in $L$.
On the other hand, if $L$ is small, the cost might be dominated by 
the $\tau^{(L)}-L$ draws from $\bar{K}$. Thus increasing $L$ might not 
significantly impact the cost until the distribution
of $\tau^{(L)} - L$ becomes stable in $L$.

The point of increasing $L$ is to obtain sharper bounds. For example, from \eqref{eq:tv_upper_bound}
we see that, for fixed $t$, the variable in the expectation takes values in $[0,1]$
with increasing probability as $L\to \infty$, resulting in upper bounds more likely to be in $[0,1]$ and thus non-vacuous. The upper bound is also decreasing in $t$.
This motivates the 
strategy of starting with \( L=1 \), plotting the bounds as in Figure \ref{TV_bounds_1}, and increasing \(L\) until the estimated upper bound
for \( d_{\text{TV}}(\pi_0, \pi) \) is close to 1. 

Irrespective of the cost, the benefits of increasing $L$ eventually diminish: the upper bounds are loose to some extent 
since the coupling operated by $\bar{K}$ is not optimal \citep{thorisson1986maximal}. The couplings considered 
in this work are chosen to be widely applicable but are not optimal in any way.

\subsection{Comparison with Johnson's diagnostics\label{sec:johnson}}

The proposed approach is similar to that proposed by Valen Johnson in \citet{valen_johnson_1996}, which works as follows.
A number $c\geq 2$ of chains start from $\pi_0$ and evolve jointly (without
time lags), such that they all coincide exactly after a random number
of steps $T_c$, while each chain marginally evolves according to $K$. 
If we assume that any draw from $\pi_0$ would be accepted as a draw from $\pi$ in a rejection sampler with probability $1-r$,
then the main result of \citet{valen_johnson_1996} provides the bound: $d_{\text{TV}}(\pi_t,\pi)\leq \mathbb{P}(T_c>t)\times(1-r^c)^{-1}$.
As $c$ increases, for any $r\in(0,1)$ the upper bound approaches $\mathbb{P}(T_c>t)$, which itself is small if $t$ 
is a large quantile of the meeting time $T_c$. 
A limitation of this result is its reliance on
the quantity $r$, which might be unknown or very close to one in challenging settings. Another difference 
is that we rely on pairs of lagged chains and tune the lag $L$, 
while the tuning parameter in \citet{valen_johnson_1996} is the number of coupled chains $c$. 

\section{ Experiments and applications } \label{section:experiments}

\subsection{Ising model} \label{section:burnin}

We consider an Ising model, where the target is
defined on a large discrete space,
% Interest in machine learning \citep{morningstar2017deep} and physics \citep{huang2017accelerated}
namely a square lattice with
$32 \times 32$ sites (each site has 4 neighbors) and periodic boundaries. For a
state $x\in \{-1,+1\}^{32\times 32}$, we define the target probability 
$\pi_\beta(x) \propto \exp (\beta \sum_{i\sim j} x_i x_j)$, where the sum is over all
pairs $i$, $j$ of neighboring sites. As $\beta$ increases, the correlation between
nearby sites increases and single-site Gibbs samplers are known to perform
poorly \citep{mossel2013exact}.  Difficulties in the assessment of the
convergence of these samplers are in part due to the discrete nature of the
state space, which limits the possibilities of visual diagnostics.  Users might
observe trace plots of one-dimensional statistics of the chains, such as
$x\mapsto \sum_{i\sim j} x_i x_j$, and declare convergence when the statistic
seems to stabilize; see \citet{titsias2017hamming,zanella2019informed} 
where trace plots of summary statistics are used to monitor
Markov chains.  

Here we compute the proposed upper bounds for the TV
distance for two
algorithms: a single site Gibbs sampler (SSG) and a parallel
tempering (PT) algorithm, where different chains target different $\pi_\beta$ 
with SSG updates, and regularly attempt to swap their states \citep{geyer1991,syed2019non}. The initial 
distribution assigns $-1$ and $+1$ with equal probability on each site independently. 
For $\beta = 0.46$, we obtain TV bounds for SSG using a lag $L=10^6$, and $N=500$ independent repeats.
For PT we use 12 chains, each targeting $\pi_\beta$ with $\beta$ in an equispaced grid ranging from $0.3$ to $0.46$,
a frequency of swap moves of $0.02$, and a lag $L=2\times 10^4$. 
The results are in Figure \ref{fig:ising:ssgversuspt}, 
where we see a plateau for the TV bounds on SSG
and faster convergence for the TV bounds on PT. 
%The results are in Figure \ref{fig:ising:ssgversuspt}, 
%where we see a plateau for the TV bounds on the SSG algorithm,
%and where the faster convergence of PT is apparent. 
Our results are consistent with theoretical work on faster mixing times of
PT targeting multimodal distributions including Ising models \citep{woodard2009}. 
Note that the targets
are different for both algorithms, as PT operates on an extended space.
The behavior of meeting times of coupled chains motivated by 
the ``coupling from the past'' algorithm \citep{propp1996exact} for Ising models has been studied
e.g. in \citet{collevecchio2018coupling}.

%Mixing time: for a grid of values of $\beta$, we estimate the iteration such that TV is less than $0.1$ and call that the mixing time.
%We plot these times as a function of $\beta$ for both SSG and PT. For each $beta$ we first obtained 1-lag meeting times,
%and then defined a new lag as twice the $99\%$ quantile of the initial meeting times.

\begin{figure}[ht!] \begin{center} 
	\includegraphics[width=0.65\paperwidth]{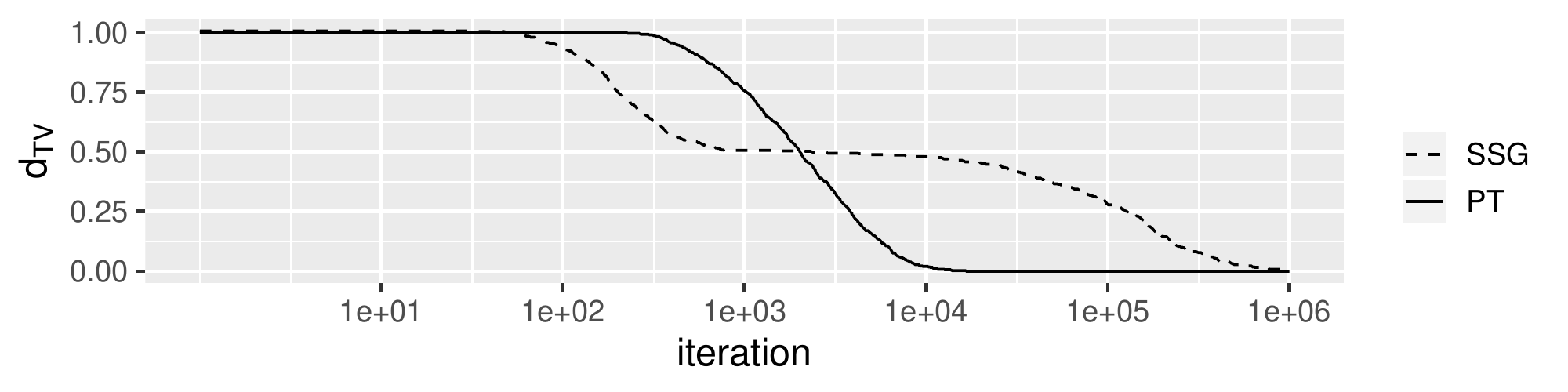}
	\caption{Single-site Gibbs (SSG) versus Parallel Tempering (PT) for an Ising model; bounds on the
    total variation distance between $\pi_t$ and $\pi$, for $t$ up to $10^6$ and inverse temperature $\beta = 0.46$.}
\label{fig:ising:ssgversuspt}
\end{center}
\end{figure}

\subsection{ Logistic regression } \label{section:exact}

We next consider a target on a continuous state space
defined as the posterior in a Bayesian logistic regression.
Consider the German Credit data from
\citet{lichman_2013}. There are $n=1000$ binary responses \( (Y_i)_{i=1}^n \in \{
-1,1 \}^n \) indicating whether individuals are creditworthy or not
creditworthy, and \(d=49\) covariates \(x_i \in \mathbb{R}^d \) for each
individual \( i \). The logistic regression model states \( \mathbb{P}(Y_i =
y_i | x_i) = (1 + e^{-y_i x_i^T \beta} )^{-1} \) with a normal prior \(
\beta \sim \mathcal{N}(0, 10 I_d) \). We can sample from the posterior using
Hamiltonian Monte Carlo (HMC, \citet{neal_1993}) or the P\'olya-Gamma
 Gibbs sampler (PG, \citet{polson_scott_windle_2013}). The former 
involves tuning parameters \(\epsilon_{\text{HMC}}\) and \(S_{\text{HMC}}\) corresponding to a step size and
a number of steps in a leapfrog integration scheme performed at every iteration. We can use the proposed 
bounds to compare convergence associated with HMC for different \(\epsilon_{\text{HMC}}, S_{\text{HMC}}\),
and with the PG sampler. Figure
\ref{fig:pg_vs_hmc_small_p} shows the total variation bounds for HMC with \(
\epsilon_{\text{HMC}} = 0.025 \) and \( S_{\text{HMC}}=4,5,6,7 \) and the corresponding bound for the
parameter-free PG sampler, both starting from \( \pi_0 \sim \mathcal{N}(0, 10 I_d) \). In this example, the 
bounds are smaller for the PG sampler than for all HMC samplers under consideration. 
%We can expect that in higher-dimensional settings
%the TV bounds would be smaller for HMC. 

We emphasize that the HMC tuning parameters associated with the fastest convergence
to stationarity might not necessarily be optimal in terms of asymptotic variance
of ergodic averages of functions of interest; see related discussions in \citet{heng_2018}. 
Also, since the proposed upper bounds are not tight, the true convergence rates of
the Markov chains under consideration may be ordered differently. The proposed upper bounds still allow 
a comparison of how confident we can be about the bias of different MCMC algorithms after a fixed number of iterations. 

%Our main message is that
%they can provide a practical, non-asymptotic comparison of algorithms.

\begin{figure} [ht!] 
\begin{center}
\includegraphics[width=0.65\paperwidth]{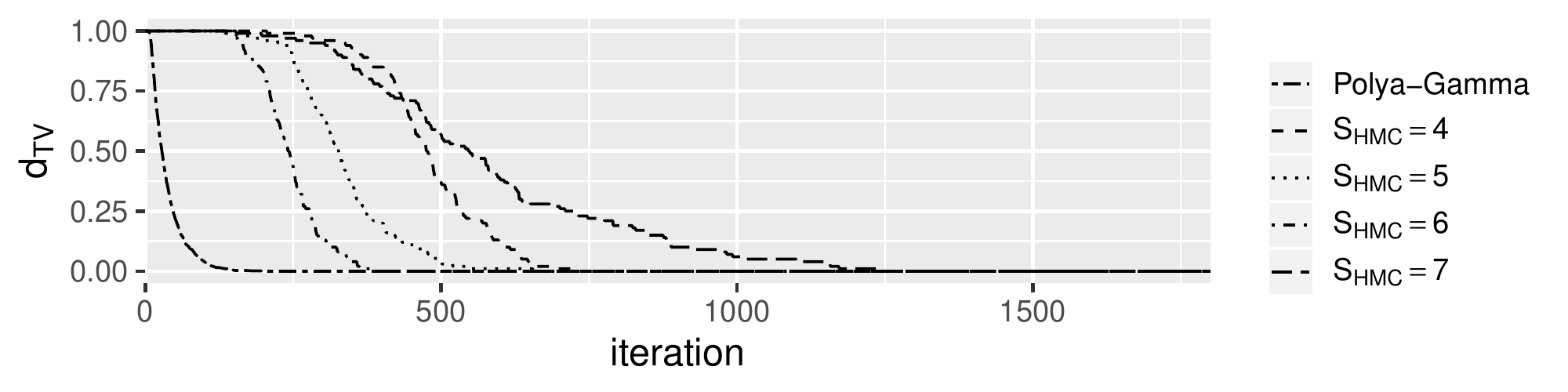}
\end{center}
\caption{Proposed upper bounds on $d_{\text{TV}}(\pi_t, \pi)$ for a P\'olya-Gamma Gibbs sampler and 
for Hamiltonian Monte Carlo on a $49$-dimensional posterior distribution in a logistic regression model. For HMC 
the step size is \( \epsilon_{\text{HMC}}=0.025 \) and the number of steps is \(S_{\text{HMC}}=4,5,6,7\).}
\label{fig:pg_vs_hmc_small_p}
\end{figure}

\subsection{ Comparison of exact and approximate MCMC algorithms } \label{section:approx}

In various settings approximate MCMC methods trade off
asymptotic unbiasedness for gains in computational speed, e.g.\ \citet{johndrow2017scalable,rudolf2018perturbation,dalalyan2019user}. 
We compare an approximate MCMC method (Unadjusted Langevin Algorithm, ULA) with its
exact counterpart (Metropolis-Adjusted Langevin Algorithm, MALA)
in various dimensions. Our target is a multivariate normal: 
\[ \pi = \mathcal{N}(0, \Sigma ) \ \text{where} \ [\Sigma]_{ i, j  } = 0.5^{|i-j|} \text{ for } 1 \leq i, j \leq d. \]
Both MALA and ULA chains start from \( \pi_0 \sim \mathcal{N}(0, I_d) \), and
have step sizes of \( d^{-1/6} \) and \( 0.1 d^{-1/6} \) respectively. Step sizes are
linked to an optimal result of \citet{roberts2001}, and the
0.1 multiplicative factor for ULA ensures that the target distribution for ULA
is close to $\pi$ (see \citet{dalalyan_2017}). We can use couplings to study the 
mixing times \( t_{\text{mix}}(\epsilon) \) of the two algorithms, where 
\( t_{\text{mix}}(\epsilon) := \inf\{ k \geq 0 : d_{\text{TV}}(\pi_k, \pi) < \epsilon \}\). 
Figure \ref{fig:ula_vs_mala_1} highlights how the dimension impacts the estimated upper bounds on
the mixing time \( t_{\text{mix}}(0.25) \), calculated as 
\( \inf\{ k \geq 0 : \widehat{\mathbb{E}} [ \max(0, \lceil (\tau^{(L)} - L -k)/L \rceil) ] < 0.25 \} \)
where $\widehat{\mathbb{E}}$ denotes empirical averages. The results are consistent with 
the theoretical analysis in \citet{dwivedi_chen_wainwright_yu_2018}. For a strongly
log-concave target such as \( \mathcal{N}(0, \Sigma ) \), Table 2 of
\citet{dwivedi_chen_wainwright_yu_2018} indicates mixing time upper bounds of order \( \mathcal{O}(d)\)
and \(\mathcal{O}(d^2) \) for ULA and MALA respectively 
(with a \textit{non-warm} start centered at the unique mode of the target). 
In comparison to theoretical studies in \citet{dalalyan_2017,
dwivedi_chen_wainwright_yu_2018}, our bounds can be directly estimated 
by simulation.
On the other hand, the bounds in \citet{dalalyan_2017,
dwivedi_chen_wainwright_yu_2018} are more explicit about the impact
of different aspects of the problem including dimension, step size, and features of the target.

\begin{figure}[ht!] 
\begin{center}
\includegraphics[width=0.65\paperwidth]{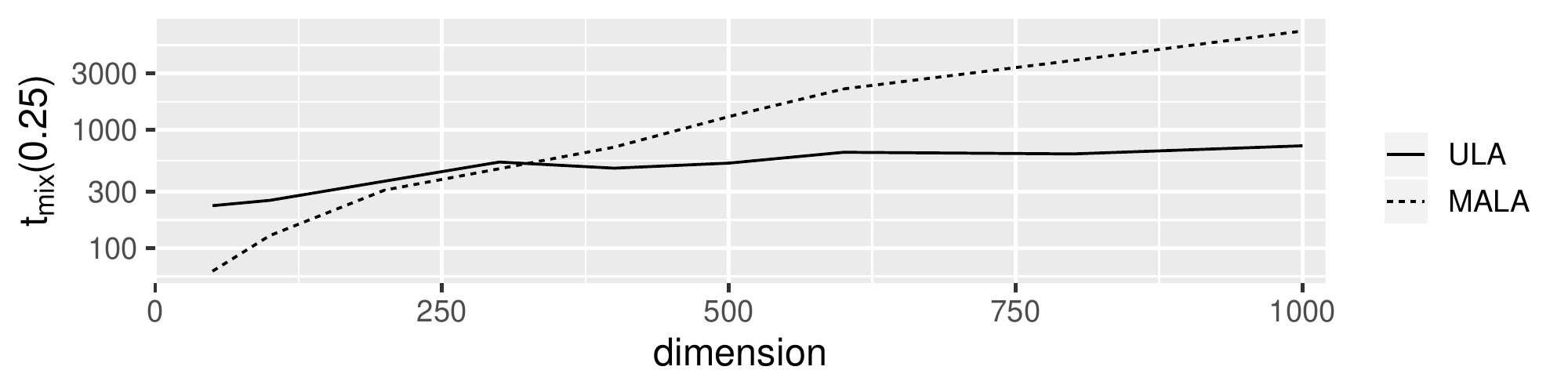}
\end{center}
\caption{Mixing time bounds for ULA and MALA targeting a multivariate Normal distribution, as a function of the dimension.
    Mixing time \( t_{\text{mix}}(0.25) \) denotes 
the first iteration $t$ for which the estimated TV between $\pi_t$ and $\pi$ is less than \(0.25\).}
\label{fig:ula_vs_mala_1}
\end{figure}

\section{Assessing the bias of sequential Monte Carlo samplers} \label{section:SMC_bound}

Lastly, we consider the bias associated with samples 
generated by sequential
Monte Carlo (SMC) samplers \citep{delmoral2006smc}; 
the bias of self-normalized importance samplers can be treated
similarly. Let $(w^n,\xi^n)_{n=1}^N$ 
be the weighted sample from an SMC sampler with \(N\) particles targeting
\(\pi\), and let \(q^{(N)}\) be the marginal distribution 
of a particle $\xi$ sampled among $(\xi^n)_{n=1}^N$ with probabilities $(w^n)_{n=1}^N$.
Our aim is to upper bound a distance between $q^{(N)}$ and $\pi$ for a fixed $N$.
We denote by $\hat{Z}$ the normalizing constant estimator generated by the SMC sampler.

The particle independent MH algorithm (PIMH, \citet{andrieu2010particle})
operates as an independent MH algorithm using SMC samplers as proposals.
Let \((\hat{Z}_{t})_{t\geq0}\) be the normalizing constant estimates from 
a PIMH chain. Consider an \(L\)-lag coupling of a pair of such PIMH chains
as introduced in \citet{pmlr-v89-middleton19a}, initializing the chains by running an SMC sampler. Here $\tau^{(L)}$ is constructed so that
it can be equal to $L$ with positive probability; more precisely, 
%The particle independent MH algorithm (PIMH, \citep{andrieu2010particle})
%operates as an independent MH algorithm using SMC samplers as proposals. 
%Consider an \(L\)-lag coupling of PIMH chains
%\citep{andrieu2010particle} as introduced in
%\citep{pmlr-v89-middleton19a}. In such setting, 
\begin{equation} \label{eq:imh_geo}
    \tau^{(L)} - (L-1) \big| \hat{Z}_{L-1} \sim \text{Geometric}( \alpha(\hat{Z}_{L-1}) ),
\end{equation} 
where \( \alpha(\hat{Z}) := \mathbb{E} \big[\min(1, \hat{Z}^*/\hat{Z})
\big| \hat{Z} \big]\) is the average acceptance probability of PIMH, from a
state with normalizing constant estimate \(\hat{Z}\); see \citet[Proposition
8]{pmlr-v89-middleton19a} for a formal statement in the case of \(1\)-lag couplings. With
this insight, we can bound the TV distance between the target and particles generated by SMC samplers, using Theorem
\ref{ipm_upper_bound} applied with $t=0$.
Details are in the appendices. We obtain 
%\begin{Prop} \label{SMC_upper_bound} Let \(q^{(N)}\) be the distribution of the weighted sample from an SMC sampler with 
%\(N\) particles targeting \(\pi\) . Let \( (\hat{Z}_{t})_{t \geq 0} \) be the unbiased estimates of 
%the normalizing constant of \(\pi\) from the corresponding PIMH Markov chain, such that marginal convergence 
%(Assumption \ref{assumption_1}) holds. Consider an \(L\)-lag coupling of a pair of such chains. Then, 
\begin{equation} \label{eq:SMC_upper_bound}
d_{\text{TV}}(q^{(N)}, \pi) \leq \mathbb{E} \Big[ \max(0, \bigl\lceil \frac{\tau^{(L)} - L}{L} \bigr\rceil) \Big] =\mathbb{E} \Big[ \frac{1- \alpha(\hat{Z}_{L-1})}{1-(1- \alpha(\hat{Z}_{L-1}))^{L}} \Big].
\end{equation}
% \end{Prop}

The bound in \eqref{eq:SMC_upper_bound} depends only on the distribution of the
normalizing constant estimator $\hat{Z}$, and can be estimated using independent
runs of the SMC sampler. 
We can also estimate the distribution of $\hat{Z}$
from a single SMC sampler by appealing to large asymptotic results
such as in \citet{berard2014lognormal}, combined with asymptotically valid variance estimators such as
\citet{lee2018variance}. As $N$ goes to infinity we expect $\alpha(\hat{Z}_{L-1})$ to approach one and the proposed upper bound to go to zero.
The proposed bound aligns with the common practice of considering the variance of $\hat{Z}$
as a measure of global performance of SMC samplers.

Existing TV bounds for particle approximations, such as those in \citet[Chapter
8]{del2004feynman} and \citet{huggins2019sequential}, are more informative
qualitatively but harder to approximate numerically.  The result also applies
to self-normalized importance samplers (see \citet[Chapter
3]{robert_casella_2013} and \citet[Chapter 8]{Owen2019mcbook}).  In that case
\citet[Theorem 2.1]{agapiou2017} shows \( d_{\text{TV}}(q^{(N)}, \pi) \leq
6N^{-1}\rho \) for \(\rho=\mathbb{E}_{\xi \sim q}[w(\xi)^2]/\mathbb{E}_{\xi
\sim q}[w(\xi)]^2 \), with $w$ the importance sampling weight function, which
is a simpler and more informative bound; see also \citet{chatterjee2018} for
related results and concentration inequalities.

%There, if we denote by $q$ the proposal distribution,
%Consider a proposal \(q\) to the target distribution
%\(\pi\), with unnormalized densities \(q^u\) and \(\pi^u \) respectively. We
%can obtain the weighted sample \((X_i, w_i)_{i=1}^N\) of size \(N\), where
%\(X_i \sim q\) are independent and \(w_i = \pi^u(X_i)/q^u(X_i)\) are
%unnormalized weights. As above, let \(q^{(N)} \) denote the distribution of
%this weighted sample. When \(\pi\) is absolutely continuous with respect to
%\(q\), SNIS estimates are consistent as \(N \rightarrow \infty\), but for
%finite \(N\) there is bias of order \(\mathcal{O}(1/N)\) \citep{agapiou2017}.
%We can assess this non-asymptotic bias by considering the PIMH chain that is
%generated by the SNIS weighted sample with \(\hat{Z} := \frac{1}{n}\sum_{i=1}^n
%w_i\) as the unbiased estimates. Under this setting, assuming that marginal
%convergence of this PIMH chain holds, Proposition {SMC_upper_bound} applies
%for SNIS.
\section{Discussion }

The proposed method can be used to obtain guidance on the choice of burn-in, to
compare different MCMC algorithms targeting the same distribution, and to
compare mixing times of approximate and exact MCMC methods.  The main
requirement for the application of the method is the ability to generate
coupled Markov chains that can meet exactly after a random but finite number of
iterations. The couplings employed here, and described in the appendices, are not optimal
in any way. As the couplings are algorithm-specific and not target-specific,
they can potentially be added to statistical software such
as PyMC3 \citep{Salvatier2016PyMC} or Stan \citep{stan2017}. 

The bounds are not tight, in part due to the couplings not being maximal
\citep{thorisson1986maximal}, but experiments suggest that they can be practical.
The proposed bounds go to zero as $t$ increases, making
them informative at least for large enough $t$.   The combination of time lags and coupling of more than two
chains as in \citet{valen_johnson_1996} could lead to new diagnostics. Further
research might also complement the proposed upper bounds with lower bounds, obtained by considering specific functions
among the classes of functions used to define the integral probability metrics.

\paragraph{Acknowledgments.} 
The authors are grateful to Espen Bernton, Nicolas Chopin, Andrew Gelman, Lester Mackey, John O'Leary, Christian Robert, 
Jeffrey Rosenthal, James Scott, Aki Vehtari and reviewers for helpful comments on an earlier version of the manuscript. The 
second author gratefully acknowledges support by the National Science Foundation through awards DMS-1712872 and DMS-1844695. 
The figures were created with packages \citep{wilke2017ggridges, wickham2016ggplot2} in R Core Team \citep{Rsoftware}. 

\bibliographystyle{abbrvnat}
\bibliography{tv_coupling_biblo.bib}

\begin{thebibliography}{64}
\providecommand{\natexlab}[1]{#1}
\providecommand{\url}[1]{\texttt{#1}}
\expandafter\ifx\csname urlstyle\endcsname\relax
  \providecommand{\doi}[1]{doi: #1}\else
  \providecommand{\doi}{doi: \begingroup \urlstyle{rm}\Url}\fi

\bibitem[Agapiou et~al.(2017)Agapiou, Papaspiliopoulos, Sanz-Alonso, and
  Stuart]{agapiou2017}
S.~Agapiou, O.~Papaspiliopoulos, D.~Sanz-Alonso, and A.~M. Stuart.
\newblock Importance sampling: Intrinsic dimension and computational cost.
\newblock \emph{Statistical Science}, 32\penalty0 (3):\penalty0 405--431, 08
  2017.
\newblock \doi{10.1214/17-STS611}.
\newblock URL \url{https://doi.org/10.1214/17-STS611}.

\bibitem[Andrieu et~al.(2010)Andrieu, Doucet, and
  Holenstein]{andrieu2010particle}
C.~Andrieu, A.~Doucet, and R.~Holenstein.
\newblock Particle {M}arkov chain {M}onte {C}arlo methods.
\newblock \emph{Journal of the Royal Statistical Society: Series B (Statistical
  Methodology)}, 72\penalty0 (3):\penalty0 269--342, 2010.

\bibitem[Andrieu et~al.(2018)Andrieu, Lee, and Vihola]{andrieu2018uniform}
C.~Andrieu, A.~Lee, and M.~Vihola.
\newblock Uniform ergodicity of the iterated conditional {SMC} and geometric
  ergodicity of particle {G}ibbs samplers.
\newblock \emph{Bernoulli}, 24\penalty0 (2):\penalty0 842--872, 2018.

\bibitem[B{\'e}rard et~al.(2014)B{\'e}rard, Del~Moral, and
  Doucet]{berard2014lognormal}
J.~B{\'e}rard, P.~Del~Moral, and A.~Doucet.
\newblock A lognormal central limit theorem for particle approximations of
  normalizing constants.
\newblock \emph{Electronic Journal of Probability}, 19, 2014.

\bibitem[Bou-Rabee et~al.(2018)Bou-Rabee, Eberle, and Zimmer]{bou2018coupling}
N.~Bou-Rabee, A.~Eberle, and R.~Zimmer.
\newblock Coupling and convergence for {H}amiltonian {M}onte {C}arlo.
\newblock \emph{arXiv preprint arXiv:1805.00452}, 2018.

\bibitem[Brooks and Roberts(1998)]{brooks1998assessing}
S.~P. Brooks and G.~O. Roberts.
\newblock Assessing convergence of {M}arkov chain {M}onte {C}arlo algorithms.
\newblock \emph{Statistics and Computing}, 8\penalty0 (4):\penalty0 319--335,
  1998.

\bibitem[Carpenter et~al.(2017)Carpenter, Gelman, Hoffman, Lee, Goodrich,
  Betancourt, Brubaker, Guo, Li, and Riddell]{stan2017}
B.~Carpenter, A.~Gelman, M.~D. Hoffman, D.~Lee, B.~Goodrich, M.~Betancourt,
  M.~Brubaker, J.~Guo, P.~Li, and A.~Riddell.
\newblock Stan : A probabilistic programming language.
\newblock \emph{Journal of Statistical Software}, 76\penalty0 (1), 1 2017.
\newblock ISSN 1548-7660.
\newblock \doi{10.18637/jss.v076.i01}.

\bibitem[Chatterjee and Diaconis(2018)]{chatterjee2018}
S.~Chatterjee and P.~Diaconis.
\newblock The sample size required in importance sampling.
\newblock \emph{Annals of Applied Probability}, 28\penalty0 (2):\penalty0
  1099--1135, 04 2018.
\newblock \doi{10.1214/17-AAP1326}.
\newblock URL \url{https://doi.org/10.1214/17-AAP1326}.

\bibitem[Chopin and Singh(2015)]{chopin2015particle}
N.~Chopin and S.~S. Singh.
\newblock On particle {G}ibbs sampling.
\newblock \emph{Bernoulli}, 21\penalty0 (3):\penalty0 1855--1883, 2015.

\bibitem[Chwialkowski et~al.(2016)Chwialkowski, Strathmann, and
  Gretton]{chwialkowski_2016_kerneltest}
K.~Chwialkowski, H.~Strathmann, and A.~Gretton.
\newblock A kernel test of goodness of fit.
\newblock In \emph{Proceedings of The 33rd International Conference on Machine
  Learning}, volume~48 of \emph{Proceedings of Machine Learning Research},
  pages 2606--2615, New York, New York, USA, 20--22 Jun 2016. PMLR.
\newblock URL \url{http://proceedings.mlr.press/v48/chwialkowski16.html}.

\bibitem[Collevecchio et~al.(2018)Collevecchio, El{\c{c}}i, Garoni, and
  Weigel]{collevecchio2018coupling}
A.~Collevecchio, E.~M. El{\c{c}}i, T.~M. Garoni, and M.~Weigel.
\newblock On the coupling time of the heat-bath process for the
  {F}ortuin--{K}asteleyn random--cluster model.
\newblock \emph{Journal of Statistical Physics}, 170\penalty0 (1):\penalty0
  22--61, 2018.

\bibitem[Corcoran and Tweedie(2002)]{corcoran2002perfectimh}
J.~N. Corcoran and R.~L. Tweedie.
\newblock Perfect sampling from independent {M}etropolis--{H}astings chains.
\newblock \emph{Journal of Statistical Planning and Inference}, 104\penalty0
  (2):\penalty0 297--314, 2002.
\newblock \doi{10.1016/S0378-3758(01)00243-9}.

\bibitem[Cowles and Rosenthal(1998)]{cowles1998simulation}
M.~K. Cowles and J.~S. Rosenthal.
\newblock A simulation approach to convergence rates for {M}arkov chain {M}onte
  {C}arlo algorithms.
\newblock \emph{Statistics and Computing}, 8\penalty0 (2):\penalty0 115--124,
  1998.

\bibitem[Dalalyan(2017)]{dalalyan_2017}
A.~S. Dalalyan.
\newblock Theoretical guarantees for approximate sampling from smooth and
  log-concave densities.
\newblock \emph{Journal of the Royal Statistical Society: Series B (Statistical
  Methodology)}, 79\penalty0 (3):\penalty0 651--676, 2017.
\newblock \doi{10.1111/rssb.12183}.
\newblock URL
  \url{https://rss.onlinelibrary.wiley.com/doi/abs/10.1111/rssb.12183}.

\bibitem[Dalalyan and Karagulyan(2019)]{dalalyan2019user}
A.~S. Dalalyan and A.~Karagulyan.
\newblock User-friendly guarantees for the {L}angevin {M}onte {C}arlo with
  inaccurate gradient.
\newblock \emph{Stochastic Processes and their Applications}, 2019.

\bibitem[Del~Moral(2004)]{del2004feynman}
P.~Del~Moral.
\newblock \emph{Feynman-{K}ac Formulae}.
\newblock Springer New York, 2004.

\bibitem[Del~Moral et~al.(2006)Del~Moral, Doucet, and Jasra]{delmoral2006smc}
P.~Del~Moral, A.~Doucet, and A.~Jasra.
\newblock Sequential {M}onte {C}arlo samplers.
\newblock \emph{Journal of the Royal Statistical Society: Series B (Statistical
  Methodology)}, 68\penalty0 (3):\penalty0 411--436, 2006.
\newblock \doi{10.1111/j.1467-9868.2006.00553.x}.
\newblock URL
  \url{https://rss.onlinelibrary.wiley.com/doi/abs/10.1111/j.1467-9868.2006.00553.x}.

\bibitem[Durmus et~al.(2016)Durmus, Fort, and Moulines]{durmus2016subgeometric}
A.~Durmus, G.~Fort, and {\'E}.~Moulines.
\newblock Subgeometric rates of convergence in {W}asserstein distance for
  {M}arkov chains.
\newblock In \emph{Annales de l'Institut Henri Poincar{\'e}, Probabilit{\'e}s
  et Statistiques}, volume~52, pages 1799--1822. Institut Henri Poincar{\'e},
  2016.

\bibitem[Dwivedi et~al.(2018)Dwivedi, Chen, Wainwright, and
  Yu]{dwivedi_chen_wainwright_yu_2018}
R.~Dwivedi, Y.~Chen, M.~J. Wainwright, and B.~Yu.
\newblock Log-concave sampling: {M}etropolis--{H}astings algorithms are fast!
\newblock In \emph{Proceedings of the 31st Conference On Learning Theory},
  volume~75 of \emph{Proceedings of Machine Learning Research}, pages 793--797.
  PMLR, 06--09 Jul 2018.

\bibitem[Gelman and Brooks(1998)]{brooks_gelman_1998}
A.~Gelman and S.~P. Brooks.
\newblock General methods for monitoring convergence of iterative simulations.
\newblock \emph{Journal of Computational and Graphical Statistics}, 1998.

\bibitem[Gelman and Rubin(1992)]{gelman_rubin_1992}
A.~Gelman and D.~B. Rubin.
\newblock Inference from iterative simulation using multiple sequences.
\newblock \emph{Statistical Science}, 1992.

\bibitem[Geweke(1998)]{geweke_1992}
J.~Geweke.
\newblock Evaluating the accuracy of sampling-based approaches to the
  calculation of pos- terior moments.
\newblock \emph{Bayesian Statistics}, 1998.

\bibitem[Geyer(1991)]{geyer1991}
C.~Geyer.
\newblock Markov chain {M}onte {C}arlo maximum likelihood.
\newblock \emph{Technical report, University of Minnesota, School of
  Statistics}, 1991.

\bibitem[Glynn and Rhee(2014)]{glynn_2014}
P.~W. Glynn and C.-H. Rhee.
\newblock Exact estimation for {M}arkov chain equilibrium expectations.
\newblock \emph{Journal of Applied Probability}, 51\penalty0 (A):\penalty0
  377--389, 2014.

\bibitem[Gorham and Mackey(2015)]{gorham_mackey_2015}
J.~Gorham and L.~Mackey.
\newblock Measuring sample quality with {S}tein's method.
\newblock In C.~Cortes, N.~D. Lawrence, D.~D. Lee, M.~Sugiyama, and R.~Garnett,
  editors, \emph{Advances in Neural Information Processing Systems 28}, pages
  226--234. Curran Associates, Inc., 2015.
\newblock URL
  \url{http://papers.nips.cc/paper/5768-measuring-sample-quality-with-steins-method.pdf}.

\bibitem[Gorham et~al.(2018)Gorham, Duncan, Vollmer, and
  Mackey]{gorham_mackey_2018}
J.~Gorham, A.~Duncan, S.~Vollmer, and L.~Mackey.
\newblock Measuring {S}ample {Q}uality with {D}iffusions.
\newblock 2018.
\newblock arXiv preprint arXiv:1611.06972v6.

\bibitem[Heng and Jacob(2019)]{heng_2018}
J.~Heng and P.~E. Jacob.
\newblock Unbiased {H}amiltonian {M}onte {C}arlo with couplings.
\newblock \emph{Biometrika}, 106\penalty0 (2):\penalty0 287--302, 2019.

\bibitem[Huggins and Roy(2019)]{huggins2019sequential}
J.~H. Huggins and D.~M. Roy.
\newblock Sequential {M}onte {C}arlo as approximate sampling: bounds, adaptive
  resampling via $\infty$-{ESS}, and an application to particle {G}ibbs.
\newblock \emph{Bernoulli}, 25\penalty0 (1):\penalty0 584--622, 2019.

\bibitem[Jacob et~al.(2019{\natexlab{a}})Jacob, Lindsten, and
  Sch\"on]{Jacob_Lindsten_Schon_2019}
P.~E. Jacob, F.~Lindsten, and T.~B. Sch\"on.
\newblock Smoothing with couplings of conditional particle filters.
\newblock \emph{Journal of the American Statistical Association},
  2019{\natexlab{a}}.
\newblock \doi{10.1080/01621459.2018.1548856}.
\newblock URL \url{https://doi.org/10.1080/01621459.2018.1548856}.

\bibitem[Jacob et~al.(2019{\natexlab{b}})Jacob, O'Leary, and
  Atchad\'e]{jacob_2019}
P.~E. Jacob, J.~O'Leary, and Y.~F. Atchad\'e.
\newblock Unbiased {M}arkov chain {M}onte {C}arlo with couplings.
\newblock \emph{Journal of the Royal Statistical Society: Series B (Statistical
  Methodology)}, 2019{\natexlab{b}}.

\bibitem[Johndrow et~al.(2018)Johndrow, Orenstein, and
  Bhattacharya]{johndrow2017scalable}
J.~E. Johndrow, P.~Orenstein, and A.~Bhattacharya.
\newblock Scalable {MCMC} for {B}ayes shrinkage priors.
\newblock \emph{arXiv preprint arXiv:1705.00841v3}, 2018.

\bibitem[Johnson(1996)]{valen_johnson_1996}
V.~E. Johnson.
\newblock Studying convergence of {M}arkov chain {M}onte {C}arlo algorithms
  using coupled sample paths.
\newblock \emph{Journal of the American Statistical Association}, 91\penalty0
  (433):\penalty0 154--166, 1996.

\bibitem[Johnson(1998)]{valen_johnson_1998}
V.~E. Johnson.
\newblock A coupling-regeneration scheme for diagnosing convergence in {M}arkov
  chain {M}onte {C}arlo algorithms.
\newblock \emph{Journal of the American Statistical Association}, 93\penalty0
  (441):\penalty0 238--248, 1998.

\bibitem[Lee and Whiteley(2018)]{lee2018variance}
A.~Lee and N.~Whiteley.
\newblock Variance estimation in the particle filter.
\newblock \emph{Biometrika}, 105\penalty0 (3):\penalty0 609--625, 2018.

\bibitem[Lichman(2013)]{lichman_2013}
M.~Lichman.
\newblock {UCI} machine learning repository, 2013.
\newblock URL http://archive.ics.uci.edu/ml.

\bibitem[Liu et~al.(2016)Liu, Lee, and Jordan]{liu2016kernelgoodnessoffit}
Q.~Liu, J.~Lee, and M.~Jordan.
\newblock A kernelized {S}tein discrepancy for goodness-of-fit tests.
\newblock In \emph{Proceedings of The 33rd International Conference on Machine
  Learning}, volume~48 of \emph{Proceedings of Machine Learning Research},
  pages 276--284, New York, New York, USA, 20--22 Jun 2016. PMLR.
\newblock URL \url{http://proceedings.mlr.press/v48/liub16.html}.

\bibitem[Mangoubi and Smith(2017)]{mangoubi2017rapid}
O.~Mangoubi and A.~Smith.
\newblock Rapid mixing of {H}amiltonian {M}onte {C}arlo on strongly log-concave
  distributions.
\newblock \emph{arXiv preprint arXiv:1708.07114}, 2017.

\bibitem[Middleton et~al.(2018)Middleton, Deligiannidis, Doucet, and
  Jacob]{middleton2018unbiased}
L.~Middleton, G.~Deligiannidis, A.~Doucet, and P.~E. Jacob.
\newblock Unbiased {M}arkov chain {M}onte {C}arlo for intractable target
  distributions.
\newblock \emph{arXiv preprint arXiv:1807.08691}, 2018.

\bibitem[Middleton et~al.(2019)Middleton, Deligiannidis, Doucet, and
  Jacob]{pmlr-v89-middleton19a}
L.~Middleton, G.~Deligiannidis, A.~Doucet, and P.~E. Jacob.
\newblock Unbiased smoothing using particle independent
  {M}etropolis-{H}astings.
\newblock In \emph{Proceedings of Machine Learning Research}, volume~89 of
  \emph{Proceedings of Machine Learning Research}, pages 2378--2387. PMLR,
  16--18 Apr 2019.
\newblock URL \url{http://proceedings.mlr.press/v89/middleton19a.html}.

\bibitem[Mossel and Sly(2013)]{mossel2013exact}
E.~Mossel and A.~Sly.
\newblock Exact thresholds for {I}sing--{G}ibbs samplers on general graphs.
\newblock \emph{The Annals of Probability}, 41\penalty0 (1):\penalty0 294--328,
  2013.

\bibitem[Neal(1993)]{neal_1993}
R.~M. Neal.
\newblock Bayesian learning via stochastic dynamics.
\newblock \emph{Advances in neural information processing systems}, 1993.

\bibitem[Owen(2019)]{Owen2019mcbook}
A.~B. Owen.
\newblock \emph{Monte Carlo theory, methods and examples}.
\newblock 2019.

\bibitem[Peyr\'e and Cuturi(2019)]{PeyreCuturi2019OptTransport}
G.~Peyr\'e and M.~Cuturi.
\newblock \emph{Computational Optimal Transport}.
\newblock 2019.
\newblock arXiv preprint ArXiv:1803.00567v3.

\bibitem[Polson et~al.(2013)Polson, Scott, and
  Windle]{polson_scott_windle_2013}
N.~G. Polson, J.~G. Scott, and J.~Windle.
\newblock Bayesian inference for logistic models using {P}olya-{G}amma latent
  variables.
\newblock \emph{Journal of the American Statistical Association}, 108\penalty0
  (504):\penalty0 1339--1349, 2013.
\newblock \doi{10.1080/01621459.2013.829001}.
\newblock URL \url{https://doi.org/10.1080/01621459.2013.829001}.

\bibitem[Propp and Wilson(1996)]{propp1996exact}
J.~G. Propp and D.~B. Wilson.
\newblock Exact sampling with coupled {M}arkov chains and applications to
  statistical mechanics.
\newblock \emph{Random Structures \& Algorithms}, 9\penalty0 (1-2):\penalty0
  223--252, 1996.

\bibitem[{R Core Team}(2013)]{Rsoftware}
{R Core Team}.
\newblock \emph{R: A Language and Environment for Statistical Computing}.
\newblock R Foundation for Statistical Computing, Vienna, Austria, 2013.
\newblock URL \url{http://www.R-project.org/}.

\bibitem[Robert and Casella(2013)]{robert_casella_2013}
C.~P. Robert and G.~Casella.
\newblock \emph{Monte Carlo Statistical Methods}.
\newblock Spinger New York, 2013.

\bibitem[Roberts and Rosenthal(2001)]{roberts2001}
G.~O. Roberts and J.~S. Rosenthal.
\newblock Optimal scaling for various {M}etropolis--{H}astings algorithms.
\newblock \emph{Statistical Science}, 16\penalty0 (4):\penalty0 351--367, 11
  2001.
\newblock \doi{10.1214/ss/1015346320}.
\newblock URL \url{https://doi.org/10.1214/ss/1015346320}.

\bibitem[Roberts and Rosenthal(2004)]{roberts_rosenthal_2004}
G.~O. Roberts and J.~S. Rosenthal.
\newblock General state space {M}arkov chains and {MCMC} algorithms.
\newblock \emph{Probability Surveys}, \penalty0 (1):\penalty0 20--71, 2004.

\bibitem[Rosenthal(1996)]{rosenthal1996analysis}
J.~S. Rosenthal.
\newblock Analysis of the {G}ibbs sampler for a model related to
  {J}ames--{S}tein estimators.
\newblock \emph{Statistics and Computing}, 6\penalty0 (3):\penalty0 269--275,
  1996.

\bibitem[Rudolf and Schweizer(2018)]{rudolf2018perturbation}
D.~Rudolf and N.~Schweizer.
\newblock Perturbation theory for {M}arkov chains via {W}asserstein distance.
\newblock \emph{Bernoulli}, 24\penalty0 (4A):\penalty0 2610--2639, 2018.

\bibitem[Salvatier et~al.(2016)Salvatier, Wiecki, and
  Fonnesbeck]{Salvatier2016PyMC}
J.~Salvatier, T.~Wiecki, and C.~Fonnesbeck.
\newblock Probabilistic programming in {P}ython using {PyMC3}.
\newblock \emph{PeerJ Computer Science}, 2\penalty0 (55), 2016.

\bibitem[Sriperumbudur et~al.(2012)Sriperumbudur, Fukumizu, Gretton,
  Sch{\"o}lkopf, and Lanckriet]{sriperumbudur2012empirical}
B.~K. Sriperumbudur, K.~Fukumizu, A.~Gretton, B.~Sch{\"o}lkopf, and G.~R.
  Lanckriet.
\newblock On the empirical estimation of integral probability metrics.
\newblock \emph{Electronic Journal of Statistics}, 6:\penalty0 1550--1599,
  2012.

\bibitem[Syed et~al.(2019)Syed, Bouchard-C{\^o}t{\'e}, Deligiannidis, and
  Doucet]{syed2019non}
S.~Syed, A.~Bouchard-C{\^o}t{\'e}, G.~Deligiannidis, and A.~Doucet.
\newblock Non-reversible parallel tempering: an embarassingly parallel {MCMC}
  scheme.
\newblock \emph{arXiv preprint arXiv:1905.02939}, 2019.

\bibitem[Thorisson(1986)]{thorisson1986maximal}
H.~Thorisson.
\newblock On maximal and distributional coupling.
\newblock \emph{The Annals of Probability}, pages 873--876, 1986.

\bibitem[Thorisson(2000)]{thorisson_2000}
H.~Thorisson.
\newblock \emph{Coupling, stationarity, and regeneration.}
\newblock Springer New York, 2000.

\bibitem[Titsias and Yau(2017)]{titsias2017hamming}
M.~K. Titsias and C.~Yau.
\newblock The {H}amming ball sampler.
\newblock \emph{Journal of the American Statistical Association}, 112\penalty0
  (520):\penalty0 1598--1611, 2017.

\bibitem[Vats et~al.(2019)Vats, Flegal, and Jones]{vats_flegal_jones_2019}
D.~Vats, J.~M. Flegal, and G.~L. Jones.
\newblock Multivariate output analysis for {M}arkov chain {M}onte {C}arlo.
\newblock \emph{Biometrika}, 106\penalty0 (2):\penalty0 321--337, 04 2019.

\bibitem[Vihola(2017)]{vihola2017unbiased}
M.~Vihola.
\newblock Unbiased estimators and multilevel {M}onte {C}arlo.
\newblock \emph{Operations Research}, 66\penalty0 (2):\penalty0 448--462, 2017.

\bibitem[Whiteley(2012)]{whiteley2012sequential}
N.~Whiteley.
\newblock Sequential {M}onte {C}arlo samplers: error bounds and insensitivity
  to initial conditions.
\newblock \emph{Stochastic Analysis and Applications}, 30\penalty0
  (5):\penalty0 774--798, 2012.

\bibitem[Wickham(2016)]{wickham2016ggplot2}
H.~Wickham.
\newblock \emph{ggplot2: elegant graphics for data analysis}.
\newblock Springer, 2016.

\bibitem[Wilke(2017)]{wilke2017ggridges}
C.~O. Wilke.
\newblock ggridges: Ridgeline plots in `ggplot2’.
\newblock \emph{R package version 0.4}, 1, 2017.

\bibitem[Woodard et~al.(2009)Woodard, Schmidler, and Huber]{woodard2009}
D.~B. Woodard, S.~C. Schmidler, and M.~Huber.
\newblock Conditions for rapid mixing of parallel and simulated tempering on
  multimodal distributions.
\newblock \emph{Ann. Appl. Probab.}, 19\penalty0 (2):\penalty0 617--640, 04
  2009.
\newblock \doi{10.1214/08-AAP555}.
\newblock URL \url{https://doi.org/10.1214/08-AAP555}.

\bibitem[Zanella(2019)]{zanella2019informed}
G.~Zanella.
\newblock Informed proposals for local {MCMC} in discrete spaces.
\newblock \emph{Journal of the American Statistical Association}, 2019.

\end{thebibliography}

\appendix

\section{Proofs}
\subsection{\texorpdfstring{$L$}{L}-lag unbiased estimators}
Our motivation for Theorem \ref{ipm_upper_bound} comes from recent works on unbiased MCMC estimators using couplings \citep{jacob_2019,glynn_2014}. In particular, extending the unbiased estimator from
\citet{jacob_2019} that corresponds to a lag $L=1$, we first construct the \( L
\)-lag estimator with an arbitrary $L \geq 1$ as 
\begin{equation} \label{eq:l_lag_estimate}
H^{(L)}_t (X,Y) := h(X_t) + \sum_{j=1}^{ \bigl\lceil \frac{\tau^{(L)} - L -t}{L} \bigr\rceil } h(X_{t+jL}) - h(Y_{t+(j-1)L}).
\end{equation}
where $h\in \mathcal{H}$, chains \((X_t)_{t \geq 0}\), \((Y_t)_{t \geq 0}\) marginally have the same
initial distribution \(\pi_0\) and Markov transition kernel \(K\) on
\((\mathbb{R}^d, \mathcal{B}(\mathbb{R}^d))\) with invariant distribution \(\pi\), and they are jointly following the \(L\)-lag coupling algorithm (Algorithm \ref{main_algo_coupling} in the main paper). As an aside, following \citet{jacob_2019} we also include the corresponding time-averaged \( L\)-lag estimator:
\begin{flalign} \label{eq:l_lag_time_average_estimate}
H^{(L)}_{k:m} (X,Y) &:= \frac{1}{m-k+1} \sum_{t=k}^m H^{(L)}_t (X,Y) \\
&= \frac{1}{m-k+1} \sum_{t=k}^m h(X_t) + \frac{1}{m-k+1} \sum_{t=k}^m \sum_{j=1}^{ \bigl\lceil \frac{\tau^{(L)} - L -t}{L} \bigr\rceil } h(X_{t+jL}) - h(Y_{t+(j-1)L}).
\end{flalign}

Following the proof technique for the \(1\)-lag estimator in \citet{jacob_2019}, we first prove an unbiasedness result for \( H^{(L)}_t (X,Y) \). By linearity the unbiasedness of \(H^{(L)}_{k:m} (X,Y)\) follows. 
\begin{Prop} \label{prop:l_lag_unbiased}
    Under the Assumptions \ref{assumption_1}, \ref{assumption_2} and \ref{assumption_3} of the main article, 
\( H^{(L)}_t (X,Y) \) has expectation \( \mathbb{E}_{X \sim \pi}[h(X)] \), finite variance, and finite expected computing time. 
\end{Prop}
\begin{proof}
    The proof is nearly identical to those in \citet{vihola2017unbiased,glynn_2014,jacob_2019} and related articles, and is only reproduced here for completeness.
Let $t=0$ without loss of generality. Otherwise start the chains at \( \pi_t \) rather than \( \pi_0 \). Secondly, we can focus on the component-wise behaviour of \( H^{(L)}_0 (X,Y) \) and assume \( h \) takes values in \(  \mathbb{R} \) without loss of generality. For simplicity of notation we drop the \( (L) \) superscript and write \( H_0 (X,Y) \) to denote  \( H^{(L)}_0 (X,Y) \). 

Define \( \Delta_0 = h(X_0) \), \( \Delta_j = h(X_{jL}) - h(Y_{(j-1)L}) \) for \( j\geq1\), and \( H^{n}_0 (X,Y) := \sum_{j=0}^n \Delta_j \). By Assumption \ref{assumption_2}, \( \mathbb{E}[\tau^{(L)}]< \infty \), so the computation time has finite expectation. When \( (1+j)L \geq \tau^{(L)}, \) \( \Delta_j = 0 \) by faithfulness (Assumption \ref{assumption_3}). As \(\tau^{(L)} \overset{a.s.}{<} \infty \), this implies \(H^{n}_0 (X,Y) \overset{a.s.}{\rightarrow} H_0 (X,Y) \) as \( n \rightarrow \infty \).

We now show that \( (H^{n}_0 (X,Y))_{n \geq 0} \) is a Cauchy sequence in \( L_2 \), the space of random variable with finite first two moments, by showing 
\[ \underset{n' \geq n}{\sup} \mathbb{E}[ \big( H_0^{n'}(X,Y) - H_0^{n}(X,Y) \big)^2 ] \underset{n \rightarrow \infty}{\rightarrow} 0. \]
This follows by direct calculation. Firstly by Cauchy--Schwarz,
\begin{flalign*}
\mathbb{E}[\big( H_0^{n'}(X,Y) - H_0^{n}(X,Y) \big)^2] = \sum_{s=n+1}^{n'} \sum_{t=n+1}^{n'} \mathbb{E}[\Delta_s \Delta_t]  \leq \sum_{s=n+1}^{n'} \sum_{t=n+1}^{n'} \mathbb{E}[\Delta_s^2]^{1/2} \mathbb{E}[\Delta_t^2]^{1/2}.
\end{flalign*}
By H\"older's inequality with \( p = 1 + \eta/2 \), \( q = (2+\eta)/\eta \) and Assumptions \ref{assumption_1} - \ref{assumption_2}, for any \( \eta > 0 \), 
\begin{flalign*}
\mathbb{E}[\Delta_t^2] = \mathbb{E}[\Delta_t^2 \textbf{1}(\tau^{(L)} > (1+t)L)] \leq \mathbb{E}[\Delta_t^{2+ \eta} ]^{\frac{1}{1+\eta/2}} \mathbb{E}[\textbf{1}(\tau^{(L)} > (1+t)L)]^{\frac{\eta}{2+\eta}} \\
< D^{\frac{1}{1+\eta/2}} (C \delta^{t})^{\frac{\eta}{2+\eta}}.
\end{flalign*}
where \( \mathbb{E}[\Delta_t^{2+ \eta} ] \leq \mathbb{E}[M_{\mathcal{H}}(X_{tL}, Y_{(t-1)L})^{2+ \eta} ] \leq D \) follows from Assumptions \ref{assumption_1}. 
Overall this implies \( \mathbb{E}[\big( H_0^{n'}(X,Y) - H_0^{n}(X,Y) \big)^2] \leq \tilde{C} \tilde{\delta}^n \) for some \(\tilde{C}>0, \tilde{\delta} \in (0,1) \) for all \( n \geq 0 \). Hence  \( (H^{n}_0 (X,Y))_{n \geq 0} \) is a Cauchy sequence in \( L_2 \), and has finite first and second moments. Recall that Cauchy sequences are bounded, so we can apply the dominated convergence theorem to get, 
\[ \mathbb{E}[H_0(X,Y)] = \mathbb{E}[ \underset{n \rightarrow \infty}{\lim} H^n_0(X,Y)] = \underset{n \rightarrow \infty}{\lim} \mathbb{E}[H^n_0(X,Y)]. \]
Finally, note that by a telescoping sum argument and Assumption \ref{assumption_1}, 
\[ \underset{n \rightarrow \infty}{\lim} \mathbb{E}[H^n_0(X,Y)] = \underset{n \rightarrow \infty}{\lim} \mathbb{E}[h(X_n)] =\mathbb{E}_{X \sim P}[h(X)]. \]
as required. Therefore, in general \( H^{(L)}_t (X,Y) \) has expectation \( \mathbb{E}_{X \sim \pi}[h(X)] \), finite variance, and a finite expected computing time.
\end{proof}

\subsection{Proof of Theorem \ref{ipm_upper_bound}}

\begin{proof} 
    We consider the \( L \)-lag estimate in \eqref{eq:l_lag_estimate}.
%\begin{equation}
%H^{(L)}_t (X,Y) := h(X_t) + \sum_{j=1}^{ \bigl\lceil \frac{\tau^{(L)} - L -t}{L} \bigr\rceil } h(X_{t+jL}) - h(Y_{t+(j-1)L}).
%\end{equation}
    Under Assumptions \ref{assumption_1}, \ref{assumption_2} and
    \ref{assumption_3}, by Proposition \ref{prop:l_lag_unbiased} \( H^{(L)}_t
    (X,Y) \) is an unbiased estimator of \( \mathbb{E}_{X \sim \pi}[ h(X) ] \),
    for any $h\in \mathcal{H}$. Then, 
\begin{flalign*}
d_{\mathcal{H}}(\pi_t, \pi) &= \sup\limits_{h \in \mathcal{H}} | \mathbb{E}_{X \sim \pi}[h(X)] - \mathbb{E}[h(X_t)] | \\
&= \sup\limits_{h \in \mathcal{H}} \Big| \mathbb{E} \Big[ \sum_{j=1}^{ \bigl\lceil \frac{\tau^{(L)} - L -t}{L} \bigr\rceil } h(X_{t+jL}) - h(Y_{t+(j-1)L}) \Big] \Big| \\
&\leq \mathbb{E} \Bigg[ \sum_{j=1}^{ \bigl\lceil \frac{\tau^{(L)} - L -t}{L} \bigr\rceil } M_{\mathcal{H}}(X_{t+jL}, Y_{t+(j-1)L}) \Bigg].
\end{flalign*}
The inequality above stems from 1) the triangle inequality applied $\lceil (\tau^{(L)} - L -t)/L \rceil$ times, and 2) the bound $|h(x)-h(y)|\leq M_\mathcal{H}(x,y)$ assumed in the main article.
We see that increasing the lag $L$ reduces the number of applications of the triangle inequality performed above, which explains the benefits of increasing $L$.
\end{proof}

\subsection{Bias of Sequential Monte Carlo samplers}

    For an SMC sampler \citep{delmoral2006smc} with \(N\) particles targeting \(\pi\), 
let $(w^n,\xi^n)_{n=1}^N$ be the particle approximation of $\pi$,
so that weighted averages $\sum_{n=1}^N w^n h(\xi^n)$ are consistent approximations of $\int h(x)\pi(dx)$ as $N\to\infty$
under some assumptions, e.g. \citet{whiteley2012sequential}.
We consider a particle $\xi$ drawn among $(\xi^n)_{n=1}^N$ with probabilities $(w^n)_{n=1}^N$,
and we denote by $q^{(N)}$ the marginal distribution of $\xi$. Our goal is to formulate an upper bound on the total variation
distance between $q^{(N)}$ and $\pi$ for fixed $N$, which is a way of studying the non-asymptotic bias of SMC samplers.

To use the proposed machinery, we embed the SMC sampler in an MCMC algorithm, following \citet{andrieu2010particle}.
The particle independent MH (PIMH) algorithm operates as follows. Initially an SMC sampler is run, 
from which a particle $\xi_0$ is drawn (marginally from $q^{(N)}$), as well as a normalizing constant estimator $\hat{Z}_{0}$ \citep{delmoral2006smc}.
We can think of the state of the chain as the pair $(\xi_{0},\hat{Z}_{0})$.
At each iteration $t\geq 1$, a new SMC sampler is run and generates $(\xi^\star, \hat{Z}^{\star})$. With probability $\min(1,\hat{Z}^\star/\hat{Z}_{t-1})$,
the new state of the chain is set to $(\xi^\star, \hat{Z}^\star)$, otherwise it remains at $(\xi_{t-1}, \hat{Z}_{t-1})$.
It is shown in \citet{andrieu2010particle} that this algorithm corresponds to 
a standard Metropolis--Hastings algorithm with independent proposals upon introducing some auxiliary variables.
Therefore under some conditions, the generated chain is such that $\xi_{t}$ goes to $\pi$ as $t\to \infty$.
We assume throughout that our three assumptions hold, which corresponds to assumptions on the performance
of the SMC sampler in the present setting.

Next consider an \( L \)-lag coupling of such a PIMH algorithm as proposed in
\citet{pmlr-v89-middleton19a} and described in Algorithm \ref{algo:pimh_coupling}. In
this setting, we can characterize the distribution of the coupling time. In
particular, 
\begin{equation} \label{eq:imh_geo:supp}
    \tau^{(L)} - (L-1) \big| \hat{Z}_{L-1} \sim \text{Geometric}( \alpha(\hat{Z}_{L-1}) ),
\end{equation}
where the Geometric distribution is parameterized to take integers values
greater than or equal to 1, and \( \alpha(\hat{Z}) := \mathbb{E} \big[\min(1,
\hat{Z}^\star/\hat{Z}) \big| \hat{Z} \big]\) is the acceptance probability
of the PIMH chain from a state with normalizing constant estimate \(\hat{Z}\).
Using a monotonicity property of IMH 
\citep{corcoran2002perfectimh}, \citet[Proposition 8]{pmlr-v89-middleton19a}
presents this result for \(1\)-Lag couplings of PIMH, and \eqref{eq:imh_geo:supp} is
a simple generalization 
to \(L\)-lag couplings; we refer to \citet{pmlr-v89-middleton19a} for the explicit assumptions being made.
Assuming that Theorem \ref{ipm_upper_bound}
applies, we consider the initial time $t=0$ and obtain
\begin{flalign*}
    d_{TV}(q^{(N)}, \pi) &\leq \mathbb{E}\Big[ \Bigl\lceil \frac{\tau^{(L)} - L }{L} \Bigr\rceil \Big] \\
&= \mathbb{E}\Big[ \mathbb{E}\Big[ \Bigl\lceil \frac{\tau^{(L)} - (L - 1) - 1}{L} \Bigr\rceil \Big| \hat{Z}_{L-1} \Big] \Big] \\
&= \mathbb{E}\Big[ \frac{1-\alpha(\hat{Z}_{L-1})}{1- (1-\alpha(\hat{Z}_{L-1}))^{L}} \Big],
\end{flalign*}
as required. 
Note that in the first inequality we used the fact that the total variation distance between some marginals of two multivariate distributions
is less than the total variation distance between the joint distributions.
The final equality follows from noting that for \( G \sim \text{Geometric}(p)\) and integers \(m \geq 0, n > 0 \), 
\begin{flalign*}
\mathbb{E}\big[ \Bigl\lceil \frac{G-m}{n} \Bigr\rceil \big] &= \sum_{k=0}^\infty \mathbb{P} \Big(\Bigl\lceil \frac{G-m}{n} \Bigr\rceil > k \Big) = \sum_{k=0}^\infty \mathbb{P} \Big(\frac{G-m}{n} > k \Big) = \frac{(1-p)^m}{1-(1-p)^n}.
\end{flalign*}

\section{Couplings of MCMC algorithms}
In this section, all the algorithms used in our examples are
presented. These are constructions used in recent work on unbiased MCMC
estimation with couplings, e.g. \citet{jacob_2019, heng_2018,pmlr-v89-middleton19a}. 
All scripts in R are available at \url{https://github.com/niloyb/LlagCouplings}. 

We first describe algorithms to sample from maximal couplings. We then describe
algorithms to sample meeting times corresponding to various couplings of MCMC algorithms.

\paragraph{Maximal Couplings.} To construct \(L\)-lag couplings, the pair of
chains needs to meet exactly whilst preserving their respective marginal
distributions. This can be achieved using \textit{maximal coupling}
\citep{valen_johnson_1998,thorisson_2000}, which we present below in Algorithm
\ref{rejection_max_coupling}. Given variables $X \sim P$, $Y\sim Q$,  Algorithm
\ref{rejection_max_coupling} samples jointly from $(X,Y)$ such that the marginal
distributions of $X$ and $Y$ are preserved and $X$ equals $Y$ with maximal probability. It
requires sampling from the distributions of $X$ and $Y$ and evaluating the ratio
of their probability density functions. Below $P$ and $Q$ denote distributions of
$X$ and $Y$; $p$ and $q$ denote the respective probability density functions.
\begin{algorithm}[!htb]
\DontPrintSemicolon
\caption{A maximal coupling of \(P\) and \(Q\)}
Sample \(X \sim P\), and \(W \sim \mathcal{U}(0,1)\)\;
\lIf{ \( p(X) W \leq q(X) \) }{
   set \(Y=X\) and \Return \( (X,Y) \)
   } \lElse {
   sample \( Y^* \sim q\) and \(W^* \sim \mathcal{U}(0,1) \) until  \( q(Y^*) W^* > p(Y^*) \). Set \(Y=Y^*\) and \Return \( (X,Y) \)
   }
\label{rejection_max_coupling}
\end{algorithm}

For the particular case when \(P = \mathcal{N}(\mu_1, \Sigma), Q =
\mathcal{N}(\mu_2, \Sigma)\), we can use a \textit{reflection-maximal coupling}
\citep{jacob_2019,bou2018coupling} which has deterministic
computational cost. This also samples jointly from $(X,Y)$ such that the
marginal distributions of $X,Y$ are preserved and $X$ equals $Y$ with maximal
probability. This is given in Algorithm \ref{max_reflection_coupling} below,
where \(s\) denotes the probability density function of a \(d\)-dimensional
standard Normal. Note that in the case \(\dot{Y}=\dot{X}+z\) below, we get an
event \(\{X=Y\}\) as required.
\begin{algorithm}[!htb] \label{max_reflection_coupling}
\DontPrintSemicolon
 \caption{A reflection-maximal coupling of $\mathcal{N}(\mu_1, \Sigma)$ and $\mathcal{N}(\mu_2, \Sigma)$}
 Let \(z= \Sigma^{-1/2}(\mu_1-\mu_2) \) and \( e=z/ \|z\| \). Sample $\dot{X} \sim \mathcal{N}(0_d, \textbf{I}_d)$, and $W \sim \mathcal{U}(0,1)$ \;
\lIf{ $ s(\dot{X}) W \leq s(\dot{X} + z) $ }{ 
	Set $\dot{Y}=\dot{X} + z $
	} \lElse {
	Set $\dot{Y}=\dot{X} -2(e^T\dot{X})e $
	}
Set \( X = \Sigma^{1/2}\dot{X} + \mu_1, Y = \Sigma^{1/2}\dot{Y} + \mu_2 \), and \Return $(X,Y)$
\end{algorithm}

When random variables \(X, Y \) have discrete distributions \( P=(p_1,\ldots,p_N), Q=(q_1,\ldots, q_N) \) on a finite state space, we can perform a maximal coupling with deterministic computation cost. This is given in Algorithm \ref{max_coupling_discrete}. First, we define $C = (c_1, \dots, c_N)$ as $c_n = (p_n \wedge q_n) / S$ for
$n\in \{1,\ldots,N\}$ with $S = \sum_{n=1}^N (p_n \wedge q_n)$.  The notation $a \wedge b$
stands for the minimum of $a$ and $b$. We then define
$P'$ and $Q'$ as $p'_n = (p_n - p_n \wedge q_n) / (1-S)$, and $q'_n = (q_n -
p_n \wedge q_n)/(1-S)$. These $P'$ and $Q'$ are probability vectors 
and computing them takes $\mathcal{O}(N)$ operations. Note that the total variation
distance between $P$ and $Q$ is equal to $1-S$, and that $P'$  and $Q'$ have disjoint supports.

\begin{algorithm}[!htb] \label{max_coupling_discrete}
\DontPrintSemicolon
\caption{A maximal coupling of \( P=(p_1,\ldots,p_N), Q=(q_1,\ldots, q_N) \)}
Sample \(U \sim \mathcal{U}(0,1) \) \;
\lIf{ \(U < S\) } { Sample \(X\) from \(C\), define \(Y=X\) and \Return $(X,Y)$ } 
\lElse { Sample \(X\) from \(P'\), \(Y\) from \(Q'\) independently,  and \Return $(X,Y)$ } 
\end{algorithm}

\subsection{ Random walk Metropolis--Hastings }
We couple a pair of random walk Metropolis--Hastings chains in Sections
\ref{section:standard_normal} and \ref{section:bimodal} using Algorithm
\ref{algo:rwmh_coupling} with step sizes \( \sigma_{\text{MH}} = 0.5 \) and \( \sigma_{\text{MH}} = 1 \) respectively.
We could also modify the algorithm to use more general proposal kernels \( q(\cdot, \cdot) \), provided that we can sample
from a maximal coupling of $q(x,\cdot)$ and $q(y,\cdot)$ for any pair $x,y$. 

\begin{algorithm}[!htb] \caption{ Gaussian random walk Metropolis--Hastings}
\DontPrintSemicolon
    \textbf{Input:} lag \( L \geq 1\), random walk step size \( \sigma_{\text{MH}} \) \;
\textbf{Output:} meeting time \( \tau^{(L)} \); chains \( (X_t)_{0 \leq t \leq \tau^{(L)} }, (Y_t)_{0 \leq t \leq \tau^{(L)}-L } \) \;
Initialize:  generate \(X_0 \sim \pi_0\) and \(Y_0 \sim \pi_0\) \;
 \For{\( t =1,\dots, L \)}{
     Sample proposal \(X^* \sim \mathcal{N}(X_{t-1}, \sigma^2_{\text{MH}}) \) \; 
     Sample \(U \sim \mathcal{U}(0,1) \) \;
\textbf{if} \(U \leq \frac{\pi(X^*)}{\pi(X_{t-1})} \), \textbf{then} set \( X_t = X^* \) ; \textbf{else} set \( X_t = X_{t-1} \)
}
 \For{\( t > L \)}{
     Sample proposals \(X^* \sim \mathcal{N}(X_{t-1}, \sigma^2_{\text{MH}}) \), \(Y^* \sim \mathcal{N}(Y_{t-1-L}, \sigma^2_{\text{MH}}) \) jointly using maximal (or reflection-maximal) coupling \; 
     Sample \(U \sim \mathcal{U}(0,1) \) \;
\textbf{if} \(U \leq \frac{\pi(X^*)}{\pi(X_{t-1})} \), \textbf{then} set \( X_t = X^* \) ; \textbf{else} set \( X_t = X_{t-1} \) \;
\textbf{if} \(U \leq \frac{\pi(Y^*)}{\pi(Y_{t-1-L})} \), \textbf{then} set \( Y_{t-L} = Y^* \) ; \textbf{else} set \(  Y_{t-L} =  Y_{t-1-L} \) \;
\smallskip
\lIf{ \( X_{t} = Y_{t-L} \) }{
   \Return \( \tau^{(L)} := t \), and the chains \( (X_t)_{0 \leq t \leq \tau^{(L)} },  (Y_t)_{0 \leq t \leq \tau^{(L)}-L } \)
   }
}
 \label{algo:rwmh_coupling}
\end{algorithm}

\subsection{ MCMC algorithms for the Ising model }
\paragraph{Single site Gibbs (SSG).} 
Our implementation of single site Gibbs (SSG) scans all the sites of the lattice systematically.
We recall that the full conditionals of the Gibbs sampling updates are Bernoulli distributed; we denote by $p(\beta,X_{-i})$
the conditional probability of site $X_i$ being equal to $+1$ given the other sites. The algorithm to sample meeting times is given in Algorithm \ref{algo:ssg}.
The SSG results in Section \ref{section:burnin} are generated using Algorithm \ref{algo:ssg} with \(\beta = 0.46\). 

\begin{algorithm}[!htb] \caption{Single Site Gibbs sampler for the Ising model }
\DontPrintSemicolon
\textbf{Input:} lag \( L \geq 1\), and inverse temperature \(\beta\)  \;
\textbf{Output:} meeting time \( \tau^{(L)} \); chains \( (X_t)_{0 \leq t \leq \tau^{(L)} }, (Y_t)_{0 \leq t \leq \tau^{(L)}-L } \) \;
Initialize:  generate \(X_0 \sim \pi_0\) and \(Y_0 \sim \pi_0\) \;
 \For{\( t =1,\dots, L \)}{
     \For{site \(i=1,\dots, 32 \times 32\)}
{
Sample \( X_{i, t} | X_{-i, t} \sim \Bernoulli(p(\beta, X_{-i, t})) \)
}
}
\For{\( t > L \)}
{
\For{site \(i=1,\dots,32 \times 32\)}
{
Sample \( X_{i, t} | X_{-i, t} \sim \Bernoulli(p(\beta, X_{-i,t})) \) and 
\( Y_{i, t-L} | Y_{-i, t-L} \sim \Bernoulli(p(\beta,Y_{-i, t-L})) \) jointly using e.g.\ Algorithm \ref{max_coupling_discrete} \; 
}
\smallskip
\lIf{ \( X_{t} = Y_{t-L} \) }{
   \Return \( \tau^{(L)} := t \), and the chains \( (X_t)_{0 \leq t \leq \tau^{(L)} },  (Y_t)_{0 \leq t \leq \tau^{(L)}-L } \)
   }
}
\label{algo:ssg}
\end{algorithm}

\paragraph{Parallel tempering (PT).} For parallel tempering, we introduce $C$ chains denoted by $x^{(1)}$, \ldots, $x^{(C)}$.
Each chain $X^{(c)}$ targets the distribution $\pi_{\beta^{(c)}}$
where $(\beta^{(c)})_{c=1}^{C}$ are positive values interpreted as inverse temperatures.
In the example in Section \ref{section:burnin}, 
we have $C=12$, $\beta^{(1)} = 0.3$, $\beta^{(C)} = 0.46$, and the intermediate
$\beta^{(c)}$ are equispaced. The frequency of proposed swap moves is denoted by $\omega$ and set to $0.02$. This is in no way optimal, see \citet{syed2019non} for 
practical tuning strategies. Our implementation of a coupled PT algorithm is given below in Algorithm \ref{algo:pt}. 

\begin{algorithm}[!htb] \caption{ Parallel tempering for the Ising model }
\DontPrintSemicolon
\textbf{Input:} lag \( L \geq 1\), and inverse temperatures \((\beta^{(c)})_{c=1}^{C}\)  \;
\textbf{Output:} meeting time \( \tau^{(L)} \), chains \( (X^{(c)}_t)_{0 \leq t \leq \tau^{(L)} }, (Y^{(c)}_t)_{0 \leq t \leq \tau^{(L)}-L } \) for \(c=1,\ldots, C\) \;
Initialize:  generate \(X^{(c)}_0 \sim \pi_0\) and \(Y^{(c)}_0 \sim \pi_0\) for each chain \(c=1,\ldots, C\) \;
 \For{\( t =1,\dots, L \)}{
	 Sample \(U \sim \mathcal{U}(0,1)\) \;
\eIf { $U < \omega$ } { 
    Define $X_{t}^{(c)} = X_{t-1}^{(c)}$ for all $c=1,\ldots, C$\\
\For{ \( c=1,\ldots, C-1 \) } {
Swap chain states \( X_{t}^{(c)}, X_{t}^{(c+1)} \) with probability \( \min \Big( 1, \frac{\pi_{\beta^{(c)}}(X_{t}^{(c+1)})\pi_{\beta^{(c+1)}}(X_{t}^{(c)})}{\pi_{\beta^{(c)}}(X_{t}^{(c)})\pi_{\beta^{(c+1)}}(X_{t}^{(c+1)})} \Big) \)
}
 }{ \For{ \( c=1,\ldots, C \) } {
 Update \(X^{(c)}_t \sim SSG( X^{(c)}_{t-1} ; \beta^{(c)}) \)
} 
   }
}
\For{\( t > L \)}{
	Sample \(U \sim \mathcal{U}(0,1)\) \;
\eIf { $U < \omega$ } { 
    Define $X_{t}^{(c)} = X_{t-1}^{(c)}$ and $Y_{t-L}^{(c)} = Y_{t-L-1}^{(c)}$ for all $c$\\
\For{ \( c=1,\ldots, C-1 \) } {
	Sample \(U^{(c)} \sim \mathcal{U}(0,1)\) \;
\textbf{if} \( U^{(c)} \leq \frac{\pi_{\beta^{(c)}}(X_{t}^{(c+1)})\pi_{\beta^{(c+1)}}(X_{t}^{(c)})}{\pi_{\beta^{(c)}}(X_{t}^{(c)})\pi_{\beta^{(c+1)}}(X_{t}^{(c+1)})}  \), swap chain states \( X_{t}^{(c)}, X_{t}^{(c+1)} \) \;
\textbf{if} \( U^{(c)} \leq \frac{\pi_{\beta^{(c)}}(Y_{t-L}^{(c+1)})\pi_{\beta^{(c+1)}}(Y_{t-L}^{(c)})}{\pi_{\beta^{(c)}}(Y_{t-L}^{(c)})\pi_{\beta^{(c+1)}}(X_{t-L}^{(c+1)})}  \), swap chain states \( Y_{t-L}^{(c)}, Y_{t-L}^{(c+1)} \) \;
}
 }{ \For{ \( c=1,\ldots, C \) } {
     Update \(X^{(c)}_t \sim SSG( X^{(c)}_{t-1} ; \beta^{(c)}) \) and \(Y^{(c)}_{t-L} \sim SSG( Y^{(c)}_{t-L-1} ; \beta^{(c)}) \) \\
 jointly using coupled SSG (see Algorithm \ref{algo:ssg}) \;
} 
   }
   \If{ \( X_t^{(c)} = Y_{t-L}^{(c)} \quad \text{for }c=1,\ldots,C\)}{
       \Return \(\tau^{(L)} := t\), and the chains \( (X_t^{(c)})_{0 \leq t \leq \tau^{(L)} },  (Y_t^{(c)})_{0 \leq t \leq \tau^{(L)}-L }\) for all $c$.
   }
}
 \label{algo:pt}
\end{algorithm}

Note that in the case of parallel tempering, meetings occur when all the $C$ pairs of chains have met.
This incurs a trade-off: increasing the number of chains might improve the performance
of the marginal algorithm but could also complicate the occurrence of meetings; see \citet{syed2019non}
for other trade-offs associated with the number of chains in parallel tempering.

\subsection{ P\'olya-Gamma Gibbs sampler}

Algorithm \ref{algo:PG} couples the P\'olya-Gamma sampler for Bayesian logistic
regression \citep{polson_scott_windle_2013}, as in Section \ref{section:exact}
with prior \( \mathcal{N}(b, B) \) on \(\beta\) for \(b=0, B = 10 I_d\). 
Parameters \( \beta,
\tilde{\beta} \in \mathbb{R}^d, W, \tilde{W} \in \mathbb{R}_+^n \) correspond
to the vectors of regression coefficients and auxiliary variables respectively for the pair of chains.
The vector \( \tilde{y} \) is defined as \(\tilde{y} = y - {1}/{2} \), where $y$ is the vector of responses \(y \in \{0,1\}^n \). 

In the algorithm, $PG(1,c)$ refers to the P\'olya-Gamma variable in the notation of \citet{polson_scott_windle_2013}.
The notation $X|rest$ refers to the conditional distribution of $X$ given all the other variables.
The tilde notation refers to components of the second chain. The coupling here was also used in \citet{jacob_2019}.

\begin{algorithm}[!htb] \caption{ P\'olya-Gamma Gibbs Coupling } 
\DontPrintSemicolon
\textbf{Input:} lag \( L \geq 1\), response \(y \in \{0,1\}^n \) and design matrix \(X \in \mathbb{R}^{n \times d} \)   \;
\textbf{Output:} meeting time \( \tau^{(L)} \); chains \( (\beta_t)_{0 \leq t \leq \tau^{(L)} }, (\tilde{\beta}_t)_{0 \leq t \leq \tau^{(L)}-L } \) \;
Initialize:  generate \(\beta_0 \sim \pi_0\) and \(\tilde{\beta}_0 \sim \pi_0\) \;
 \For{\( t =1,\dots, L \)}{
Sample \( W_{t,i} | rest \sim PG(1, |x_i^T \beta_{t-1}|) \) for \(i=1,\dots,n\) \;
Sample \( \beta_t | rest \sim \mathcal{N}(\Sigma(W_t)(X^T\tilde{y} + B^{-1}b), \Sigma(W_t) ) \) \ for \ \( \Sigma(W_t) = (X^T \text{diag}(W_t) X +  B^{-1})^{-1} \) \;
}
 \For{\( t > L \)}{
Sample \(  W_{t,i} | rest \) and \(  \tilde{W}_{t-L,i} | \tilde{rest} \), jointly using maximal couplings
of $PG(1, |x_i^T \beta_{t-1}|)$ and $PG(1, |x_i^T \tilde{\beta}_{t-L-1}|)$,
for $i=1,\ldots,n$, 
by noting that the ratio of density functions of two P\'olya-Gamma random variables is tractable: \\
\[ \forall x > 0, \ \frac{PG(x;1,c_1)}{PG(x;1,c_2)} = \frac{\text{cosh}(c_2/2)}{\text{cosh}(c_1/2)} \exp \Big( - \Big( \frac{c_2^2}{2} - \frac{c_1^2}{2} \Big)x \Big) \]
Sample \( \beta_t | rest\) and \( \tilde{\beta}_{t-L} | \tilde{rest}\) from a maximal
coupling of \\
\[ \mathcal{N}(\Sigma(W_t)(X^T\tilde{y} + B^{-1}b), \Sigma(W_t) ) \text{ and } \mathcal{N}(\Sigma(\tilde{W}_{t-L})(X^T\tilde{y} + B^{-1}b), \Sigma(\tilde{W}_{t-L}) ) \]
\lIf{ \( \beta_{t} = \tilde{\beta}_{t-L} \) }{
   \Return \( \tau^{(L)} := t \), and the chains \( (\beta_t)_{0 \leq t \leq \tau^{(L)} },  (\tilde{\beta}_t)_{0 \leq t \leq \tau^{(L)}-L } \).
   }

}
\label{algo:PG}
\end{algorithm}

\subsection{ Hamiltonian Monte Carlo }

Algorithm \ref{algo:HMC} couples Hamiltonian Monte Carlo (HMC), as used in Section \ref{section:exact}. We follow the coupling construction from \citet{heng_2018};
see also references therein. For simplified notation, we will use \(
K_{p}(\beta, \cdot \ ; \ \epsilon_{\text{HMC}}, S_{\text{HMC}} ) \) to denote the leapfrog integration 
and the accept-reject part of HMC from position \(\beta \in
\mathbb{R}^d \) with momentum \(p \in \mathbb{R}^d \). 
Here \(\epsilon_{\text{HMC}}\) and \(S_{\text{HMC}}\) correspond to the step size and
the number of steps respectively in the leapfrog integration scheme.

We write
\(\bar{K}_{\text{RWMH}}((\beta,\tilde{\beta}), \cdot \ ; \ \sigma_{\text{MH}} ) \) to denote the kernel
of the coupled random walk 
Metropolis--Hastings algorithm (Algorithm \ref{algo:rwmh_coupling}) with step size \(\sigma_{\text{MH}}\). Mixture parameter
\(\gamma\) corresponds to the probability of selecting kernel \( \bar{K}_{\text{RWMH}}((\beta,\tilde{\beta}), \cdot \ ; \ \sigma_{\text{MH}} ) \)
from a mixture of the kernels \(K_{p}(\beta, \cdot \ ; \ \epsilon_{\text{HMC}}, S_{\text{HMC}} )\) and
\( \bar{K}_{\text{RWMH}}((\beta,\tilde{\beta}), \cdot \ ; \ \sigma_{\text{MH}} ) \).
The HMC results in Section \ref{section:exact} are generated using Algorithm \ref{algo:HMC} with 
\( \epsilon_{\text{HMC}} = 0.025, S_{\text{HMC}} = 4,5,6,7 \), \( \gamma = 0.05 \) and \(\sigma_{\text{MH}} = 0.001 \).

\begin{algorithm}[!htb] \caption{ Hamiltonian Monte Carlo}
\DontPrintSemicolon
\textbf{Input:} lag \( L \geq 1\), mixture parameter \(\gamma \in (0,1) \), and random walk step size \( \sigma_{\text{MH}} \)  \;
\textbf{Output:} meeting time \( \tau^{(L)} \); chains \( (\beta_t)_{0 \leq t \leq \tau^{(L)} }, (\tilde{\beta}_t)_{0 \leq t \leq \tau^{(L)}-L } \) \;
Initialize:  generate \(\beta_0 \sim \pi_0\) and \(\tilde{\beta}_0 \sim \pi_0\) \;
 \For{\( t =1,\dots, L \)}{
Sample momentum \( p^* \sim \mathcal{N}(0_d, \textbf{I}_d) \) and sample \( \beta_t \sim K_{p^*}(\beta_{t-1}, \cdot \ ; \ \epsilon_{\text{HMC}}, S_{\text{HMC}} ) \) \;
}
 \For{\( t > L \)}{
	 Sample \( U \sim \mathcal{U}(0,1) \) \;

\eIf{\( U \leq \gamma \)}{ 
    Sample \( \beta_t, \tilde{\beta}_{t-L} \sim \bar{K}_{\text{RWMH}}((\beta_{t-1},\tilde{\beta}_{t-L-1}), \cdot \ ; \ \sigma_{\text{MH}})  \) using Algorithm \ref{algo:rwmh_coupling} \;
 } {
 Sample common momentum \( p^* \sim \mathcal{N}(0_d, \textbf{I}_d) \) \;
 Sample \( \beta_t \sim K_{p^*}(\beta_{t-1}, \cdot \ ; \ \epsilon_{\text{HMC}}, S_{\text{HMC}} ) \) and \( \tilde{\beta}_{t-L} \sim K_{p^*}(\tilde{\beta}_{t-1-L}, \cdot \ ; \ \epsilon_{\text{HMC}}, S_{\text{HMC}} ) \) \;
 }
 \smallskip
\lIf{ \( \beta_{t} = \tilde{\beta}_{t-L} \) }{
    \Return \( \tau^{(L)} := t \), and the chains \( (\beta_t,W_t)_{0 \leq t \leq \tau^{(L)} },  (\tilde{\beta}_t, \tilde{W}_t)_{0 \leq t \leq \tau^{(L)}-L } \)
   }

}
\label{algo:HMC}
\end{algorithm}

Note that reflection-maximal coupling can also be used to draw the momenta in coupled Hamiltonian Monte Carlo,
as discussed in \citet{bou2018coupling,heng_2018}.

\subsection{ Metropolis--adjusted Langevin Algorithm }
The Metropolis--adjusted Langevin Algorithm (MALA) can be coupled as in random walk Metropolis--Hastings,
as it corresponds to a particular choice of proposal distribution.
For simplicity of notation we use 
\( q_{\sigma}(X, \cdot) \sim \mathcal{N}(X + \frac{1}{2}\sigma^2 \nabla \log \pi(X), \sigma^2 \textbf{I}_d ) \) to denote the Langevin proposal.
The MALA results in Section \ref{section:approx} are generated using Algorithm \ref{algo:mala_coupling} with 
\(\sigma = d^{-1/6}\) for \(d = 50, 100, 200, 300, 400, 500, 600, 800, 1000 \).

\begin{algorithm}[!htb] \caption{ MALA }
\DontPrintSemicolon
\textbf{Input:} lag \( L \geq 1\), random walk step size \( \sigma \) \;
\textbf{Output:} meeting time \( \tau^{(L)} \); chains \( (X_t)_{0 \leq t \leq \tau^{(L)} }, (Y_t)_{0 \leq t \leq \tau^{(L)}-L } \) \;
Initialize:  generate \(X_0 \sim \pi_0\) and \(Y_0 \sim \pi_0\) \;
 \For{\( t =1,\dots, L \)}{
Sample proposal \(X^* \sim q_{\sigma}(X_{t-1}, \cdot) \) \; 
Sample \(U \sim \mathcal{U}(0,1) \) \;
\textbf{if} \(U \leq \frac{\pi(X^*) q_{\sigma}(X^*, X_{t-1}) }{\pi(X_{t-1}) q_{\sigma}(X_{t-1}, X^*) } \), \textbf{then} set \( X_t = X^* \) ; \textbf{else} set \( X_t = X_{t-1} \)
}
 \For{\( t > L \)}{
Sample proposals \(X^* \sim q_{\sigma}(X_{t-1}, \cdot) \), \(Y^* \sim q_{\sigma}(Y_{t-1-L}, \cdot) \) jointly via reflection-maximal coupling of Algorithm \ref{max_reflection_coupling} \; 
Sample \(U \sim \mathcal{U}(0,1) \) \;
\textbf{if} \(U \leq \frac{\pi(X^*) q_{\sigma}(X^*, X_{t-1}) }{\pi(X_{t-1}) q_{\sigma}(X_{t-1}, X^*) } \), \textbf{then} set \( X_t = X^* \) ; \textbf{else} set \( X_t = X_{t-1} \) \;
\textbf{if} \(U \leq \frac{\pi(Y^*) q_{\sigma}(Y^*, Y_{t-1-L}) }{\pi(Y_{t-1-L}) q_{\sigma}(Y_{t-1-L}, Y^*) } \), \textbf{then} set \( Y_{t-L} = Y^* \) ; \textbf{else} set \(  Y_{t-L} =  Y_{t-1-L} \) \;
\smallskip
\lIf{ \( X_{t} = Y_{t-L} \) }{
   \Return \( \tau^{(L)} := t \), and the chains \( (X_t)_{0 \leq t \leq \tau^{(L)} },  (Y_t)_{0 \leq t \leq \tau^{(L)}-L } \)
   }
}
\label{algo:mala_coupling}
\end{algorithm}

\subsection{ Unadjusted Langevin Algorithm }
Unadjusted Langevin proceeds as MALA but without the MH acceptance step. 
Thus an algorithm to sample meeting times for coupled ULA chains follows from the algorithm described for coupled MALA algorithm, simply by removing the acceptance steps. As before, we use 
\( q_{\sigma}(X, \cdot) \sim \mathcal{N}(X + \frac{1}{2}\sigma^2 \nabla \log \pi(X), \sigma^2 \textbf{I}_d ) \) to denote the Langevin proposal.
The ULA results in Section \ref{section:approx} are generated using Algorithm \ref{algo:ula_coupling} with 
\(\sigma = 0.1d^{-1/6}\) for \(d = 50, 100, 200, 300, 400, 500, 600, 800, 1000 \).

\begin{algorithm}[!htb] \caption{ ULA }
\DontPrintSemicolon
\textbf{Input:} lag \( L \geq 1\), random walk step size \( \sigma \)  \;
\textbf{Output:} meeting time \( \tau^{(L)} \); chains \( (X_t)_{0 \leq t \leq \tau^{(L)} }, (Y_t)_{0 \leq t \leq \tau^{(L)}-L } \) \;
Initialize:  generate \(X_0 \sim \pi_0\) and \(Y_0 \sim \pi_0\) \;
 \For{\( t =1,\dots, L \)}{
     Sample $X_t \sim q_{\sigma}(X_{t-1},\cdot)$
%Sample proposal \(X^* \sim q_{\sigma}(X_{t-1}, \cdot) \) \; 
%Sample \(U \sim Unif(0,1) \) \;
%\textbf{if} \(U \leq \frac{\pi(X^*) q_{\sigma}(X^*, X_{t-1}) }{\pi(X_{t-1}) q_{\sigma}(X_{t-1}, X^*) } \), \textbf{then} set \( X_t = X^* \) ; \textbf{else} set \( X_t = X_{t-1} \).
}
 \For{\( t > L \)}{
Sample \(X_t \sim q_{\sigma}(X_{t-1}, \cdot) \), \(Y_{t-L} \sim q_{\sigma}(Y_{t-1-L}, \cdot) \) jointly via reflection-maximal coupling of Algorithm \ref{max_reflection_coupling} \;
\smallskip 
\lIf{ \( X_{t} = Y_{t-L} \) }{
   \Return \( \tau^{(L)} := t \), and the chains \( (X_t)_{0 \leq t \leq \tau^{(L)} },  (Y_t)_{0 \leq t \leq \tau^{(L)}-L } \)
   }
}
\label{algo:ula_coupling}
\end{algorithm}

\subsection{Particle independent Metropolis--Hastings}

By construction, \( \tau^{(L)}>L \) almost surely for all the above couplings.
Here we describe a version of coupled particle independent
Metropolis--Hastings (PIMH) which allows coupling at the first step, such that
\( \tau^{(L)}=L \) can occur with positive probability. This coupling was 
introduced in \citet{pmlr-v89-middleton19a}. 

\begin{algorithm}[!htb] \caption{Particle independent Metropolis--Hastings}
\DontPrintSemicolon
    \textbf{Input:} lag \( L \geq 1\), and SMC sampler targeting \(\pi\) \;
\textbf{Output:} meeting time \( \tau^{(L)} \); chains \( (\xi_t, Z_t)_{0 \leq t \leq \tau^{(L)} }, (\tilde{\xi}_t, \tilde{Z}_t)_{0 \leq t \leq \tau^{(L)}-L } \) where \(Z_t,\tilde{Z}_t \) are unbiased estimates of the normalizing constant of \(\pi\) \;
Initialize: Sample $\xi_0, Z_0$ from the SMC sampler \;
 \For{\( t =1,\dots, (L-1) \)}{
     Sample proposal \(\xi^*, Z^*\) from the SMC sampler \; 
     Sample \(U \sim \mathcal{U}(0,1) \) \;
\textbf{if} \(U \leq \frac{Z^*}{Z_{t-1}} \), \textbf{then} set \( \xi_t = \xi^*, Z_t=Z^* \) ; \textbf{else} set \( \xi_t = \xi_{t-1}, Z_t=Z_{t-1} \)
}
\For{\( t = L \)}{
     Sample proposal \(\xi^*, Z^*\) from the SMC sampler \; 
     Sample \(U \sim \mathcal{U}(0,1) \) \;
\textbf{if} \(U \leq \frac{Z^*}{Z_{L-1}} \), \textbf{then} set \( \xi_L = \xi^*, Z_L=Z^* \) ; \textbf{else} set \( \xi_L = \xi_{L-1}, Z_L=Z_{L-1} \) \;
Set \( \tilde{\xi}_0 = \xi^*, \tilde{Z}_0 = Z^* \)
}
\For{\( t > L \)}{
 Sample proposal \(\xi^*, Z^* \) from the SMC sampler \;
 Sample $U \sim \mathcal{U}(0,1)$ \;
\textbf{if} \(U \leq \frac{Z^*}{Z_{t-1}} \), \textbf{then} set \( \xi_t = \xi^*, Z_t=Z^* \) ; \textbf{else} set \( \xi_t = \xi_{t-1}, Z_t=Z_{t-1} \) \;
\textbf{if} \(U \leq \frac{Z^*}{\tilde{Z}_{t-L-1}} \), \textbf{then} set \( \tilde{\xi}_{t-L} = \xi^*, \tilde{Z}_{t-L}=Z^* \) ; \textbf{else} set \( \tilde{\xi}_{t-L} = \tilde{\xi}_{t-L-1}, \tilde{Z}_{t-L}=\tilde{Z}_{t-L-1} \)\;
\If{ \( \xi_{t} = \tilde{\xi}_{t-L}, Z_t=\tilde{Z}_{t-L} \) }{
   \Return \( \tau^{(L)} := t \), and the chains \( (\xi_t, Z_t)_{0 \leq t \leq \tau^{(L)} },  (\tilde{\xi}_t, \tilde{Z}_t)_{0 \leq t \leq \tau^{(L)}-L } \).
   }
}
\label{algo:pimh_coupling}
\end{algorithm}

\end{document}